\documentclass[11pt,a4paper]{article}
\usepackage{times,amsmath,amssymb,amsthm}
\usepackage{bm,marvosym}
\usepackage{color,booktabs}
\usepackage{pgf,pgfarrows,pgfnodes}
\usepackage{tikz,pifont,pgflibraryplotmarks}
\usepgflibrary{shapes.symbols}
\usepackage{fancybox,xkeyval,xcolor,geometry}
\usepackage[pdftex]{hyperref}
\hypersetup{plainpages=True, pdfstartview=FitV,
colorlinks=true,linkcolor=blue,citecolor=blue}

\def\v#1{{\mathbf{#1}}}
\def\mb#1{{\mathbb{#1}}}

\newcommand{\norm}[1]{\left|\!\left|#1\right|\!\right|_1}

\newtheorem{cor}{Corollary}
\newtheorem{thm}{Theorem}
\newtheorem{lmm}{Lemma}

\title{\bf Simple, efficient maxima-finding algorithms
for \\ multidimensional samples}
\author{{\sc Wei-Mei Chen}\\
    Department of Electronic Engineering\\
    National Taiwan University of Science and Technology\\
    Taipei 106\\
    Taiwan
\and {\sc Hsien-Kuei Hwang}\\
    Institute of Statistical Science\\
    Academia Sinica\\
    Taipei 115\\
    Taiwan
\and {\sc Tsung-Hsi Tsai}\\
    Institute of Statistical Science\\
    Academia Sinica\\
    Taipei 115\\
    Taiwan
}
\date{\today}

\begin{document}
\maketitle

\begin{abstract}
New algorithms are devised for finding the maxima of
multidimensional point samples, one of the very first problems
studied in computational geometry. The algorithms are very simple
and easily coded and modified for practical needs. The expected
complexity of some measures related to the performance of the
algorithms is analyzed. We also compare the efficiency of the
algorithms with a few major ones used in practice, and apply our
algorithms to find the maximal layers and the longest common
subsequences of multiple sequences.
\end{abstract}


\section{Introduction}

A point $\v{p}\in \mb{R}^d$ is said to \emph{dominate} another point
$\v{q}\in\mb{R}^d$ if the coordinatewise difference $\v{p}-\v{q}$
has only nonnegative coordinates and $\v{p}-\v{q}$ is not
identically a zero vector, where the dimensionality $d\ge1$. For
convenience, we write $\v{q}\prec \v{p}$ or $\v{p}\succ \v{q}$. The
non-dominated points in a sample are called the \emph{maxima} or
\emph{maximal points} of that sample. Note that there may be two
identical points that are both maxima according to our definition of
dominance. Since there is no total order for multidimensional points
when $d>1$, such a dominance relation among points has been one of
the simplest and widely used partial orders. We can define dually
the corresponding \emph{minima} of the sample by reversing the
direction of the dominance relation.

\subsection{Maxima in diverse scientific disciplines}

Daily lives are full of tradeoffs or multi-objective decision
problems with often conflicting factors; the numerous terms appeared
in different scientific fields reveal the importance and popularity
of maxima in theory, algorithms, applications and practice: maxima
(or vector maxima) are sometimes referred to as \emph{nondominance},
\emph{records}, \emph{outer layers}, \emph{efficiency}, or
\emph{noninferiority} but are more frequently known as \emph{Pareto
optimality}  or \emph{Pareto efficiency} (with the natural
derivative \emph{Pareto front}) in econometrics, engineering,
multi-objective optimization, decision making, etc. Other terms used
with essentially the same denotation include \emph{admissibility} in
statistics, \emph{Pareto front} (and the corresponding notion of
\emph{elitism}) in evolutionary algorithms, and \emph{skyline} in
database language; see \cite{BHLT01,CHT03,Devroye83,Devroye93} and
the references therein and the books \cite{CVL07,Deb01,Ehrgott00}
for more information. They also proved useful in many computer
algorithms and are closely related to several well-known problems,
including convex hulls, top-$k$ queries, nearest-neighbor search,
largest empty rectangles, minimum independent dominating set in
permutation graphs, enclosure problem for rectilinear $d$-gon,
polygon decomposition, visibility and illumination, shortest path
problem, finding empty slimplices, geometric containment problem,
data swapping, grid placement problem, and multiple longest common
subsequence problem to which we will apply our algorithms later; see
\cite{CHT03,PS85} for more references.

We describe briefly here the use of maxima in the contexts of
database language and multi-objective optimization problems using
evolutionary algorithms.

Skylines in database queries are nothing but minima. A typical
situation where the skyline operator arises is as follows; see
\cite{BKS01} for details. Travelers are searching over the Internet
for cheap hotels that are near the beach in, say C\^ote d'Azur.
Since the two criteria ``lower price" and ``shorter distance" are
generally conflicting with each other and since there are often too
many hotels to choose from, one is often interested in those hotels
that are non-dominated according to the two criteria; here dominance
is defined using minima. Much time will be saved if the search or
sort engine can automatically do this and filter out those that are
dominated for database queriers (by, say clicking at the skyline
operator). On the other hand, frequent spreadsheet users would also
appreciate such an operator, which can find the maxima, minima or
skyline of multidimensional data by simple clicks.

In view of these and many other natural applications such as
e-commerce, multivariate sorting and data visualization, the
skylines have been widely and extensively addressed in recent
database literature, notably for low- and moderate-dimensional data,
following the pioneering paper \cite{BKS01}. In addition to devising
efficient skyline-finding algorithms, other interesting issues
include top-$k$ representatives, progressiveness, absence of false
hits, fairness, incorporation of preference, and universality. A
large number of skyline-finding algorithms have been proposed for
various needs; see, for example, \cite{BCP08,BKS01, GSG07,
KRR02, PTFS05, TEO01} and the references therein.

On the other hand, an area receiving even much more recent attention
is the study of multi-objective evolutionary algorithms (MOEAs),
where the idea of maxima also appeared naturally in the form of
non-dominated solutions (or elites). MOEAs provide a popular
approach for multi-objective optimization, which identify the most
feasible solutions lying on the Pareto front under various (often
conflicting) constraints by repeatedly finding non-dominated
solutions based on biological evolutionary mechanisms. These
algorithms have turned out to be extremely fruitful in diverse
engineering, industrial and scientific areas, as can be witnessed by
the huge number of citations many papers on MOEA have received so
far. Some popular schemes in this context suggested the maintenance
of an explicit archive/elite for all non-dominated solutions found
so far; see below and \cite{FES03,KC00,ZDT00,ZT99} and the
references therein. See also \cite{CC06} for an interesting
historical overview.

Finally, maxima also arises in a random model for river
networks (see \cite{BLP06,BR04}) and in an interesting statistical
estimate called ``layered nearest neighbor estimate" (see
\cite{BD08}).

\subsection{Maxima, maximal layers and related notions}

Maxima are often used for some ranking purposes or used as a
component problem for more sophisticated situations. Whatever the
use, one can easily associate such a notion to define
multidimensional sorting procedures. One of the most natural ways is
to ``peel off'' the current maxima, regarded as the first-layer
maxima, and then finding the maxima of the remaining points,
regarded then as the second-layer maxima, and so on until no point
is left. The total number of such layers gives rise to a natural
notion of \emph{depth}, which is referred to as the \emph{height} of
the corresponding random, partially ordered sets in \cite{BW88}.
Such a maximal-layer depth is nothing but the length of the longest
increasing subsequences in random permutations when the points are
selected uniformly and independently from the unit square, a problem
having attracted widespread interests, following the major
breakthrough paper \cite{BDJ99}.

On the other hand, the maximal layers are closely connected to
chains (all elements comparable) and antichains (all elements
incomparable) of partially ordered set in order theory, an
interesting result worthy of mention is the following dual version
of Dilworth's theorem, which states that the size of the largest
chain in a finite partial order is equal to the smallest number of
antichains into which the partial order may be partitioned; see, for
example, \cite{Kaldewaij87} for some applications.

In addition to these aspects, \emph{maximal layers} have also been
widely employed in multi-objective optimization applications since
the concept was first suggested in Goldberg's book \cite{G89}. Based
on identifying the maximal layer one after another, Srinivas and Deb
\cite{SD95} proposed the non-dominated sorting genetic algorithm
(NSGA) to simultaneously find multiple Pareto-optimal points, which
was later on further improved in \cite{Deb02}, reducing the time
complexity from $O(dn^3)$ to $O(dn^2)$. (This paper has soon become
highly-cited.) Jensen \cite{Jensen03} then gave a divide-and-conquer
algorithm to find the maximal layers with time complexity $(n(\log
n)^{d-1})$; see Section~\ref{sec-app} for more details.

In the contexts of multi-objective optimization problems, elitism
usually refers to the mechanism of storing some obtained
non-dominated solutions into an external archive during the process
of MOEAs because a non-dominated solution with respect to its
current data is not necessarily non-dominated with respect to the
whole feasible solutions. The idea of elitism was first introduced
in \cite{ZT99} and is regarded as a milestone in the development of
of MOEAs \cite{CC06}. Since the effectiveness of this mechanism
relies on the size of the external non-dominated set, an elite
archive with limited size was suggested to store the truncated
non-dominated sets \cite{KC00,ZT99}, so as to avoid the
computational costs of maintaining all non-dominated sets.
Nevertheless, restricting the size of archive reduces the quality of
solutions; more efficient storages and algorithms are thus studied
for unconstrained elite archives; see for example
\cite{FES03,Jensen03,MTT02}.

\subsection{Aim and organization of this paper}

Due to the importance of maxima, a large number of algorithms for
finding them in a given sample of points have been proposed and
extensively studied in the literature, and many different design
paradigms were introduced including divide-and-conquer, sequential,
bucket or indexing, selection, and sieving; see \cite{CHT03} for a
brief survey. Quite naturally, practical algorithms often merge more
than one of the design paradigms for better performance.

Despite the huge number of algorithms proposed in the literature,
there is still need of simpler and practically efficient algorithms
whose performance does not deteriorate too quickly in massive point
samples as the number of maximal points grows, a property which we
simply refer to as ``scalable". This is an increasingly important
property as nowadays massive data sets or data streams are becoming
ubiquitous in diverse areas.

Although for most practical ranking and selecting purposes, the
notion of maxima is most useful when the number of maxima is not too
large compared with the sample size, often there is no a priori
information on the number of maxima before computing them.

Furthermore, a general-purposed algorithm may in practice face the
situation of data samples with very large standard deviation for
their maxima. From known probabilistic theory of maxima (see
\cite{BDHT05} and the references therein), the expected number of
maxima and the corresponding variance can in two typical random
models grow either in $O((\log n)^{d-1})$ when the coordinates are
roughly independent or in $O(n^{1-1/d})$ when the coordinates are
roughly negatively dependent, both $O$-terms here referring to large
$n$, the sample size, and fixed $d$, the dimensionality. In
particular, in the planar case, there can be $\sqrt{n}$ number of
maxima on average for roughly negatively correlated coordinates, in
contrast to $\log n$ for independent coordinates; see also
\cite{Baryshnikov07,Golin93} for the ``gap theorem" and
\cite{Devroye94} for a similar $\sqrt{n}$ vs $\log n$ effect
(reflecting dependence or independence) on random Cartesian trees.
Since the maximal points can be very abundant with large standard
deviations, more efficient and more uniformly scalable algorithms
are needed.

We propose in this paper two new techniques to achieve scalability:
the first technique is to reduce the maxima-finding to a two-phase
records-finding procedure, giving rise to a no-deletion algorithm,
which largely simplifies the design and maintenance of the data
structure used. The second technique is the introduction of bounding
box in the corresponding tree structure for storing the current
maxima, which reduces significantly the deterioration of efficiency
in higher dimensions. The combined use of both techniques on $k$-d
trees turns out to be very efficient, easily coded and outperforms
many known efficient algorithms. Some preliminary results on the use
of $k$-d trees for finding maxima of appeared in \cite{CL07}.

This paper is organized as follows. In the next section, we briefly
describe some existing algorithms proposed in the diverse
literature, focusing on the two most popular and representative
paradigms: divide-and-conquer and sequential.
Section~\ref{sec-new-algo} gives details of the new techniques,
implementation on $k$-d trees, and diverse aspects of further
improvements. A comparative discussion will also be given with major
known algorithms. Analytic and empirical aspects of the performance
of the algorithms will be discussed in Section~ \ref{sec-analysis}.
Finally, we apply our new algorithm to the problems of finding
maximal layers and that of finding multiple longest common
subsequence in Section~\ref{sec-app}, where the efficiency of our
algorithm is tested on several data sets.

\emph{Throughout this paper, $\v{Max}(\v{p})$ always denotes the
maxima of the sequence of points $\v{p}=\{\v{p}_1,\dots,\v{p}_n\}$.}

\section{Known maxima-finding algorithms---a brief account}

In view of the large amount of algorithms with varying characters
appeared in the literature, it is beyond the scope of this paper to
provide a full description of all existing algorithms. Instead, we
give a brief account here on divide-and-conquer and sequential
algorithms; see \cite{CHT03} and the references there for other
algorithms.

\subsection{Divide-and-conquer algorithms}

Divide-and-conquer algorithms were first proposed by Kung et al.\
\cite{KLP75} with the worst-case time complexity of order $n(\log
n)^{d-2 +\delta_{d,2}}$ for dimensionality $d\ge2$, where $n$ is the
number of points and $\delta_{a,b}$ denotes the Kronecker delta
function. Bentley \cite{Bentley80} schematized a multidimensional
divide-and-conquer paradigm, which in particular is applicable to
the maxima-finding problem with the same worst-case complexity.
Gabow et al.\ \cite{GBT84} later improved the complexity to
$O(n(\log n)^{d-3}\log \log n)$ for $d\ge4$ by scaling techniques.
Output-sensitive algorithms with complexity of order $n(\log
(M+1))^{d-2 +\delta_{d,2}}$ were devised in \cite{KS85}, where $M$
denotes the number of maxima.

The typical pattern of most of these algorithms is as follow.
\begin{quote}
\textbf{Algorithm}
\textsf{Divide-and-conquer} \\
//\textbf{Input}: A sequence of points $\v{p}=
\{\v{p}_1,\ldots, \v{p}_n\}$ in $\mb{R}^d$\\
//\textbf{Output}: $\v{Max}(\v{p})$\\
\textbf{begin}\\%
\hspace*{.5cm} \textbf{if } $n\le 1$ \textbf{then}
    \textbf{return}$(\{\v{p}_1,\dots,\v{p}_n\})$ \\%
\hspace*{.5cm} \textbf{else} \textbf{return}
\textsf{Filter-out-false-maxima}(\textsf{Divide-and-conquer}($\{
\v{p}_1,\dots,\v{p}_{\lfloor n/2\rfloor} \}$),\\
\hspace*{5.9cm} \textsf{Divide-and-conquer}($\{\v{p}_{\lfloor
n/2\rfloor+1},\dots,\v{p}_n\}$)\\
\textbf{end}
\end{quote}
Here \textsf{Filter-out-false-maxima}($\v{p},\v{q}$) drops maxima in
$\v{q}$ that are dominated by maxima in $\v{p}$.

These divide-and-conquer algorithms are generally characterized by
their good theoretic complexity in the worst case, simple structural
decompositions in concept but low competitiveness in practical and
typical situations with sequential algorithms, although it is known
that most divide-and-conquer algorithms have linear expected-time
performance under the usual hypercube random model, or more
generally when the expected number of maxima is of order
$o(n^{1-\varepsilon})$; see \cite{Devroye83,FG93}. Variants of them
have however been adapted in the skyline and evolutionary
computation contexts; see for example \cite{PTFS05} for skylines and
\cite{Jensen03} for MOEAs.

\subsection{Sequential algorithms}

The most widely-used procedure for finding non-dominated points in
multidimensional samples has the following incremental, on-line,
one-loop pattern (see \cite{BCL93,KLP75}).

\begin{quote}
\textbf{Algorithm} \textsf{Sequential}\\
//\textbf{Input}: A sequence of points
$\v{p}=\{\v{p}_1,\dots,
\v{p}_{n}\}$ in $\mb{R}^d$\\
//\textbf{Output}: $\v{Max}(\v{p})$\\
\textbf{begin}\\%
\hspace*{1cm} $\v{M}:=\{\v{p}_1\}$\qquad
//$\v{M}:$ a data structure for storing the current maxima\\
\hspace*{1cm} \textbf{for} $i:=2$ \textbf{to} $n$ \textbf{do}\\
\hspace*{2cm} \textbf{if} no point in $\v{M}$ dominates
    $\v{p}_i$ \textbf{then} \qquad //updating $\v{M}$\\
\hspace*{3cm} delete $\{\mathbf{q\in M\,:\,q\prec p}_i\}$ from
    $\v{M}$\\%
\hspace*{3cm} insert $\v{p}_i$ into $\v{M}$ \\
\textbf{end}
\end{quote}

The algorithm is a natural adaptation of the one-dimensional
maximum-finding loop, which represents the very first algorithm
analyzed in details in Knuth's \emph{Art of Computer Programming}
books \cite{Knuth97}. It runs by comparing points one after another
with elements in the data structure $\v{M}$, which stores the maxima
of all elements seen so far, called \emph{left-to-right maxima} or
\emph{records}; it moves on to the next point $\v{p}_{i+1}$ if the
new point $\v{p}_i$ is dominated by some element in $\v{M}$, or it
removes elements in $\v{M}$ dominated by the new point $\v{p}_i$ and
accepts the new point $\v{p}_i$ into $\v{M}$.

For dimensions $d\ge2$, such a simple design paradigm was first
proposed in \cite{KLP75} (with an additional pre-sorting stage for
one of the coordinates) and the complexity was analyzed for $d=2$
and $d=3$. To achieve optimal worst-case complexity for $d=3$, they
used AVL-tree (a simple, balanced variant of binary search tree).
The simpler implementation using a linear list (and without any
pre-sorting procedure) was discussed first in the little known paper
\cite{Habenicht83} and later in greater detail in \cite{BCL93}, in
particular with the move-to-front self-adjusting heuristic.

The \textsf{Sequential} algorithm, also known as block-nested-loop
algorithm \cite {PTFS05}, is most efficient when the number of
maxima is a small function of $n$ such as powers of logarithm, but
deteriorates rapidly when the number of maxima is large. In addition
to list employed in \cite{BCL93,Habenicht83} to store the maxima for
sequential algorithms, many varieties of tree structures were also
proposed in the literature: quad trees in \cite{Habenicht83,MTT02},
R-trees in \cite{KRR02}, and $d$-ary trees in \cite{Schutze03}; see
also \cite{PTFS05}. But these algorithms become less efficient (in
time bound and in space utilization) as the dimensionality of data
increases, also the maintenance is more complicated. We will see
that the use of $k$-d trees is preferable in most cases; see also
\cite{CHT03} for the use of binary search trees for $d=2$.

\section{A two-phase sequential algorithm based on $k$-d trees using
bounding boxes} \label{sec-new-algo}

We present our new algorithm based on the ideas of multidimensional
non-dominated records, bounding boxes, and $k$-d trees. Further
refinements of the algorithm will also be discussed. We then compare
our algorithm with a few major ones discussed in the literature.

\subsection{The design techniques}

We introduce in this subsection multidimensional non-dominated
records, $k$-d trees and bounding boxes, and will apply them later
for finding maxima. In practice, each of these techniques can be
incorporated equally well into other techniques for finding maxima.

\subsubsection{Multidimensional non-dominated records}

Except for simple data structures such as list, the deletion
performed in algorithm \textsf{Sequential} is often the most
complicated step as it requires a structure re-organization after
the removal of the dominated elements. It is then natural to see if
there are algorithms avoiding or reducing deletions.

Note that in the special case when $d=1$, the two steps ``deletion"
and ``insertion" in algorithm \textsf{Sequential} actually reduce to
one, and the inserted elements are nothing but the records (or
record-breaking elements, left-to-right maxima, outstanding
elements, etc.). Recall that an element $\v{p}_j$ in the sequence of
reals $\{\v{p}_1,\dots,\v{p}_n\}$ is called a record if $\v{p}_j$ is
not dominated by any element in $\{\v{p}_1,\dots,\v{p}_{j-1}\}$.

The crucial observation is then based on extending the
one-dimensional records to higher dimensions.

\smallskip

\noindent \textbf{Definition ($d$-dimensional non-dominated
records).} A point $\v{p}_j$ in the sequence of points in $\mb{R}^d$
$\{\v{p}_1,\dots, \v{p}_n\}$ is said to be a \emph{$d$-dimensional
non-dominated record} of the sequence $\{\v{p}_1,\ldots ,\v{p}_n\}$
if $\v{p}_j$ is not dominated by $\v{p}_i$ for all $1\le i<j$. We
also define $\v{p}_1$ to be a non-dominated record.

\smallskip

Such non-dominated records are called ``weak records" in
\cite{Gnedin07}, but this term seems less informative; see also
\cite{Devroye93} for a different use of records. \emph{For
simplicity, we write, throughout this paper, records to mean
non-dominated records when no ambiguity will arise.}

For convenience, write $\v{Rec}(\v{p})$ as the set of records of
${\v{p}}=\{\v{p}_1,\ldots, \v{p}_n\}$.

\begin{lmm} For any given set of points
$\{\v{p}_1,\dots,\v{p}_n\}$, \label{lmm}
\[
    \v{Max}(\{\v{p}_1,\dots,\v{p}_n\})
    = \v{Rec}(\overline{\v{Rec}
    (\{\v{p}_1,\dots,\v{p}_n\})}),
\]
where $\overline{\{\v{q}_1,\dots,\v{q}_k\}}
:=\{\v{q}_k,\dots,\v{q}_1\}$ denotes the reversed sequence.
\end{lmm}
In words, if $\{\v{q}_1,\dots ,\v{q}_k\}$ represents the records of
the sequence $\{\v{p}_1,\dots ,\v{p}_n\}$, then the maxima of
$\{\v{p}_1,\dots ,\v{p}_n\}$ is equal to the records of the sequence
$\{\v{q}_k,\v{q}_{k-1},\dots ,\v{q}_1\}$.
\begin{proof}
We prove by contradiction. Assume that there are two points
$\v{p}_i$ and $\v{p}_j$ in the set
\[
    \v{Rec}(\overline{\v{Rec}
    (\{\v{p}_1,\dots,\v{p}_n\})})
\]
such that $\v{p}_i\succ\v{p}_j$. If $i<j$, then $\v{p}_j$ cannot be
a record and thus cannot be a member of the set
$\v{Rec}(\{\v{p}_1,\dots,\v{p}_n\})$, a contradiction. On the other
hand, if $i>j$, then $ \v{p}_j$ is a record and is included in the
set $\v{Rec}(\{\v{p}_1,\dots,\v{p}_n\})$, but then after the order
being reversed, it cannot be a record since it is dominated by
$\v{p}_i$, again a contradiction.
\end{proof}

Another interesting property regarding the connection between
records and maxima is the following.
\begin{cor} \label{cor1} In algorithm \textsf{Sequential}
for finding maxima, the points $\v{p}_i$ to be inserted in the
for-loop are necessarily the records, while those deleted are
records but not maxima.
\end{cor}

\subsubsection{A two-phase sequential algorithm}

Lemma~\ref{lmm} provides naturally a two-phase, no-deletion
algorithm for finding maxima: in the first phase, we identify the
records, and in the second phase, we find the records of the
reversed sequence of the output of the first phase (so as to remove
the records that are not maxima); an example of seven planar points
is given in Figure~\ref{ex1}. In other terms, we perform only the
insertion in algorithm \textsf{Sequential} in the first phase,
postponing the deletion to be carried out in the second.

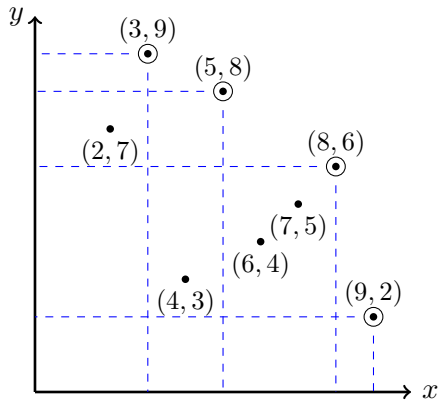
\begin{figure}[tbp]
\begin{minipage}{0.48\textwidth}
\centering
\begin{tikzpicture} [scale=0.5]
\draw [->,line width=1pt] (0,0)--(10,0) node[right]{$x$};%
\draw [->,line width=1pt] (0,0)--(0,10) node[left]{$y$};%
\foreach \position in {(2,7),(4,3),(7,5),(6,4)}%
\fill \position circle (0.1cm)node[below,black]{\small$\position$};
\draw [dashed,blue](9,1.5)--(9,0) (8.5,2)--(0,2);
\draw [dashed,blue](8,5.5)--(8,0) (7.5,6)--(0,6);
\draw [dashed,blue](5,7.5)--(5,0) (4.5,8)--(0,8);
\draw [dashed,blue](3,8.5)--(3,0) (2.5,9)--(0,9);
\foreach \position in {(9,2),(8,6),(5,8),(3,9)}
\fill \position circle (0.1cm) node[above]{\small$\position$};
\foreach \position in {(9,2),(8,6),(5,8),(3,9)}
\draw \position circle (0.25cm);
\end{tikzpicture}
\end{minipage}
\begin{minipage}{0.5\textwidth}
\centering
\begin{tikzpicture} [line width=1pt,scale=0.5]%
\draw (-0.5,8) node[right]{Input};%
\draw (0.8,7) node[minimum width=1cm,draw]{$2,7$}
      (2.8,7) node[minimum width=1cm,draw]{$3,9$}
      (4.8,7) node[minimum width=1cm,draw]{$4,3$}
      (6.8,7) node[minimum width=1cm,draw]{$5,8$}
      (8.8,7) node[minimum width=1cm,draw]{$7,5$}
      (10.8,7) node[minimum width=1cm,draw]{$6,4$}
      (12.8,7) node[minimum width=1cm,draw]{$8,6$}
      (14.8,7) node[minimum width=1cm,draw]{$9,2$};
\draw (0.8,6) node[minimum width=1cm]{\scriptsize$1$}
      (2.8,6) node[minimum width=1cm]{\scriptsize$2$}
      (4.8,6) node[minimum width=1cm]{\scriptsize$3$}
      (6.8,6) node[minimum width=1cm]{\scriptsize$4$}
      (8.8,6) node[minimum width=1cm]{\scriptsize$5$}
      (10.8,6) node[minimum width=1cm]{\scriptsize$6$}
      (12.8,6) node[minimum width=1cm]{\scriptsize$7$}
      (14.8,6) node[minimum width=1cm]{\scriptsize$8$};
\draw (-0.5,5) node[right]{Phase 1: processing input data from left
to right};%
\draw (0.8,4) node[minimum width=1cm,draw,fill=yellow]{$2,7$}
      (2.8,4) node[minimum width=1cm,draw]{$3,9$}
      (4.8,4) node[minimum width=1cm,draw,fill=yellow]{$4,3$}
      (6.8,4) node[minimum width=1cm,draw]{$5,8$}
      (8.8,4) node[minimum width=1cm,draw,fill=yellow]{$7,5$}
      (10.8,4) node[minimum width=1cm,draw]{$8,6$}
      (12.8,4) node[minimum width=1cm,draw]{$9,2$};
\draw (0.8,3) node[minimum width=1cm]{\scriptsize$1$}
      (2.8,3) node[minimum width=1cm]{\scriptsize$2$}
      (4.8,3) node[minimum width=1cm]{\scriptsize$3$}
      (6.8,3) node[minimum width=1cm]{\scriptsize$4$}
      (8.8,3) node[minimum width=1cm]{\scriptsize$5$}
      (10.8,3) node[minimum width=1cm]{\scriptsize$6$}
      (12.8,3) node[minimum width=1cm]{\scriptsize$7$};
\draw (-0.5,2) node[right]{Phase 2: processing the list in Phase 1
from right to left}; %
\draw (0.8,1) node[minimum width=1cm,draw]{$9,2$}
      (2.8,1) node[minimum width=1cm,draw]{$8,6$}
      (4.8,1) node[minimum width=1cm,draw]{$5,8$}
      (6.8,1) node[minimum width=1cm,draw]{$3,9$};
\draw (0.8,0) node[minimum width=1cm]{\scriptsize$1$}
      (2.8,0) node[minimum width=1cm]{\scriptsize$2$}
      (4.8,0) node[minimum width=1cm]{\scriptsize$3$}
      (6.8,0) node[minimum width=1cm]{\scriptsize$4$};
\end{tikzpicture}
\end{minipage}
\caption{The maxima of the point sample
$\{(2,7),(3,9),(4,3),(5,8),(7,5),(6,4),(8,6),(9,2)\}$ are marked by
circles. After Phase 1, $(2,7)$, $(4,3)$ and $(7,5)$ are still left
in the list though they are not maximal points. But after Phase 2,
the resulting list contains all maximal points.} \label{ex1}
\end{figure}

The precise description of the algorithm is given as follows. Note
that in the algorithm a list $\v{R}$ is used to store the records
and has to preserve their relative orders.

\medskip

\textbf{Algorithm} \textsf{Two-Phase}\\
\hspace*{1cm}//\textbf{Input}: A sequence of points
$\v{p}=\{\v{p}_1,\dots, \v{p}_{n}\}$\\
\hspace*{1cm}//\textbf{Output}: $\v{Max}(\v{p})$\\ %
\hspace*{1cm}\textbf{begin}\\%
\hspace*{1.2cm}// Phase 1\\%
\hspace*{2cm} $\v{R}:=\{\v{p}_1\}$ \qquad //
$\v{R}$ stores the non-dominated records \\%
\hspace*{2cm} $k:=1$ \qquad // $k$ counts the number of
records\\%
\hspace*{2cm} \textbf{for} $i:=2$ \textbf{to} $n$ \textbf{do}
\hspace*{2cm}\\%
\hspace*{3cm} \textbf{if} $\v{p}_i$ is not dominated by any
point in $\v{R}$ \textbf{then} \\%
\hspace*{4cm} $k:=k+1$\\%
\hspace*{4cm} insert $\v{p}_i$ at the end of $\v{R}$
\hspace*{1.2cm}// so as to retain the input order \\%
\hspace*{1.2cm}// After the for-loop, $\v{R} =\{\v{p}_{j_1},\dots,
\v{p}_{j_k}\}$,
where $j_1<j_2<\dots<j_k$.\\%
\hspace*{1.2cm}// Phase 2 \\%
\hspace*{2cm} $\v{M}:=\{\v{p}_{j_k}\}$\qquad //
$\v{M}$ stores the maxima \\%
\hspace*{2cm} \textbf{for} $i:=k-1$ \textbf{downto} $1$ \textbf{do}
\\ %
\hspace*{3cm} \textbf{if} $\v{p}_{j_i}$ is not dominated by any
point in $\v{M}$ \textbf{then}
insert $\v{p}_{j_i}$ in $\v{M}$ \\
\hspace*{1cm}\textbf{end}

The correctness of Algorithm \textsf{Two-Phase} is guaranteed by
Lemma~\ref{lmm}.

While the two-phase procedure may increase the total number of
comparisons made, the real scalar comparisons made can actually be
simplified since we need only to detect if the incoming element is
dominated by some element in the list $\v{R}$, and there is no need
to check the reverse direction that the incoming element dominates
some element in $\v{R}$. Thus the code for the detection of
dominance or non-dominance is simpler than that of algorithms
performing deletions. Furthermore, for each vector comparison, it is
not necessary to check all coordinates unless one element is
dominated by the other. Briefly, the two-phase algorithm splits the
comparisons made for checking dominance between elements in two
directions.

\subsubsection{The $k$-d trees}

The data structure $k$-d tree (or multidimensional binary search
tree) is a natural extension of binary search tree for
multidimensional data, where $k$ denotes the dimensionality. For
more notational convenience and consistency, we also write,
throughout this paper, $d$ as the dimensionality (but still use
$k$-d tree instead of $d$-$d$ tree). It was first invented by
Bentley \cite{Bentley75}. The idea is to use each of the $d$
coordinates cyclically at successive levels of the tree as the
\emph{discriminator} and direct points falling in the subtrees. If a
node holding the point $\v{r} =(r_1,\dots ,r_d)$ in a $k$-d tree has
the $\ell$-th coordinate as the discriminator, then, for any node
holding the point $\v{w}=(w_1,\dots ,w_d)$ in the subtrees of
$\v{r}$, we have the relation $w_\ell<r_\ell$ if $\v{w}$ lies in the
left-subtree of $\v{r}$, $w_\ell\ge r_\ell$ if $\v{w}$ lies in the
right-subtree of $\v{r}$. The children of $\v{r}$ then move on to
the $(\ell\ \text{mod}\ d)+1$-st coordinate as the discriminator. A
two-dimensional example is given in Figure~\ref{kd}.

\begin{figure}[!htbp]
\begin{center}
\begin{tabular}{|l|l|}
\hline
\begin{tikzpicture}[scale=.8]
\node (a) at (0.5,3.5) {\ding{172}}; \draw (0,0) rectangle (3, 3);
\draw (1,0) -- (1, 3); \fill [black] (1,1.5) circle (1.5pt)
node[above right] {{\tiny $\v{p}_1$}};
\end{tikzpicture}
\begin{tikzpicture}[scale=.8,level distance=8mm,
every node/.style={draw=black,rectangle,inner sep=2pt}]
\draw[draw=none] (0,0) rectangle (4, 4); \node[draw=black,
style=circle] (a) at (2,3) {\tiny{$\v{p}_1$}}
    child[sibling distance=20mm]{node{}}
    child[sibling distance=20mm]{node{}}
;
\end{tikzpicture}
&
\begin{tikzpicture}[scale=.8]
\node (a) at (0.5,3.5) {\ding{173}}; \draw (0,0) rectangle (3, 3);
\draw (1,0) -- (1, 3); \draw (1,2.2) -- (3, 2.2); \fill [black]
(1,1.5) circle (1.5pt) node[above right] {{\tiny $\v{p}_1$}}; \fill
[black] (2.5,2.2) circle (1.5pt) node[above right] {{\tiny
$\v{p}_2$}};
\end{tikzpicture}
\begin{tikzpicture}[scale=.8,level distance=8mm,
every node/.style={draw=black,rectangle,inner sep=2pt}]
\draw[draw=none] (0,0) rectangle (4, 4); \node[draw=black,
style=circle] (a) at (2,3) {\tiny{$\v{p}_1$}}
    child[sibling distance=20mm]{node{}}
    child[sibling distance=20mm]{node[draw=black,style=circle]
    {{\tiny $\v{p}_2$}}
        child[sibling distance=10mm]{node{}}
        child[sibling distance=10mm]{node{}}
    }
;
\end{tikzpicture}
\\ \hline
\begin{tikzpicture}[scale=.8]
\node (a) at (0.5,3.5) {\ding{174}}; \draw (0,0) rectangle (3, 3);
\draw (1,0) -- (1, 3); \draw (1,2.2) -- (3, 2.2); \draw (1.7,0) --
(1.7, 2.2); \fill [black] (1,1.5) circle (1.5pt) node[above right]
{{\tiny $\v{p}_1$}}; \fill [black] (2.5,2.2) circle (1.5pt)
node[above right] {{\tiny $\v{p}_2$}}; \fill [black] (1.7,1) circle
(1.5pt) node[above right] {{\tiny $\v{p}_3$}};
\end{tikzpicture}
\begin{tikzpicture}[scale=.8,level distance=8mm,
every node/.style={draw=black,rectangle,inner sep=2pt}]
\draw[draw=none] (0,0) rectangle (4, 4); \node[draw=black,
style=circle] (a) at (2,3) {\tiny{$\v{p}_1$}}
    child[sibling distance=20mm]{node{}}
    child[sibling distance=20mm]{node[draw=black,style=circle]
    {{\tiny $\v{p}_2$}}
        child[sibling distance=10mm]{node[draw=black,style=circle]
        {{\tiny $\v{p}_3$}}
            child[sibling distance=10mm]{node{}}
            child[sibling distance=10mm]{node{}}
        }
        child[sibling distance=10mm]{node{}}
    }
;
\end{tikzpicture}
&
\begin{tikzpicture}[scale=.8]
\node (a) at (0.5,3.5) {\ding{175}}; \draw (0,0) rectangle (3, 3);
\draw (1,0) -- (1, 3); \draw (1,2.2) -- (3, 2.2); \draw (1.7,0) --
(1.7, 2.2); \draw (0,2.0) -- (1, 2.0); \fill [black] (1,1.5) circle
(1.5pt) node[above right] {{\tiny $\v{p}_1$}}; \fill [black]
(2.5,2.2) circle (1.5pt) node[above right] {{\tiny $\v{p}_2$}};
\fill [black] (1.7,1) circle (1.5pt) node[above right] {{\tiny
$\v{p}_3$}}; \fill [black] (.5,2) circle (1.5pt) node[above right]
{{\tiny $\v{p}_4$}};
\end{tikzpicture}
\begin{tikzpicture}[scale=.8,level distance=8mm,
every node/.style={draw=black,rectangle,inner sep=2pt}]
\draw[draw=none] (0,0) rectangle (4, 4); \node[draw=black,
style=circle] (a) at (2,3) {\tiny{$\v{p}_1$}}
    child[sibling distance=20mm]{node[draw=black,style=circle]
    {{\tiny $\v{p}_4$}}
        child[sibling distance=10mm]{node{}}
        child[sibling distance=10mm]{node{}}
    }
    child[sibling distance=20mm]{node[draw=black,style=circle]
    {{\tiny $\v{p}_2$}}
        child[sibling distance=10mm]{node[draw=black,style=circle]
        {{\tiny $\v{p}_3$}}
            child[sibling distance=10mm]{node{}}
            child[sibling distance=10mm]{node{}}
        }
        child[sibling distance=10mm]{node{}}
    }
;
\end{tikzpicture}
\\\hline
\end{tabular}
\end{center}
\caption{The stepwise construction of a $2$-$d$ tree of four
points.} \label{kd}
\end{figure}
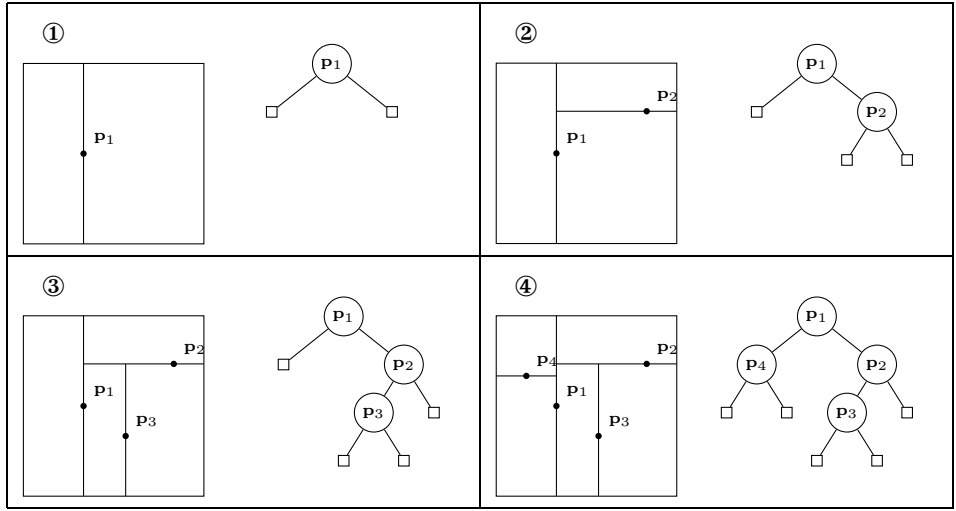

\subsubsection{Bounding-boxes}

Bounding boxes are simple techniques in improving the performance of
many algorithms, especially those dealing with intersecting
geometric objects, and have been widely used in many theoretical and
practical situations.

The application of bounding boxes is straightforward. Let
$\v{u}_{\v{r}}=(u_1,\dots,u_d)$, where $u_i$ is the maximum among
all the $i$-th coordinates of points in the subtree rooted at
$\v{r}$. Then $\v{u}_{\v{r}}$ is defined to be the \emph{upper
bound} of the subtree rooted at $\v{r}$ or simply the upper bound of
the node $\v{r}$. Similarly, define $\v{v}_{\v{r}}=(v_1,\dots,v_d)$
to be the \emph{lower bound} of the subtree rooted at $\v{r}$, where
$v_i$ is the minimum among all the $i$-th coordinates of points in
the subtree rooted at $\v{r}$. A simple example of three-dimensional
points is given in Figure~\ref{bound}. For simplicity, we also use
the upper (or lower) bound of a node . The upper and lower bounds of
a node constitute a bounding box for that subtree.

Now if a point $ \v{p}$ is not dominated by $\v{u}_{\v{r}}$, then
obviously $\v{p}$ is not dominated by any point in the subtree
rooted at $\v{r}$. This means that all comparisons between $\v{p}$
and all points in the subtree rooted at $\v{r}$ can be avoided.
Similarly, when searching for points in the subtree rooted at
$\v{r}$ that are dominated by $\v{p}$, we can first compare it with
$\v{v}_{\v{r}}$, and all comparisons between $\v{p}$ with each node
of that subtree can be saved if $\v{v}_{\v{r}}$ is not dominated by
$\v{p}$.

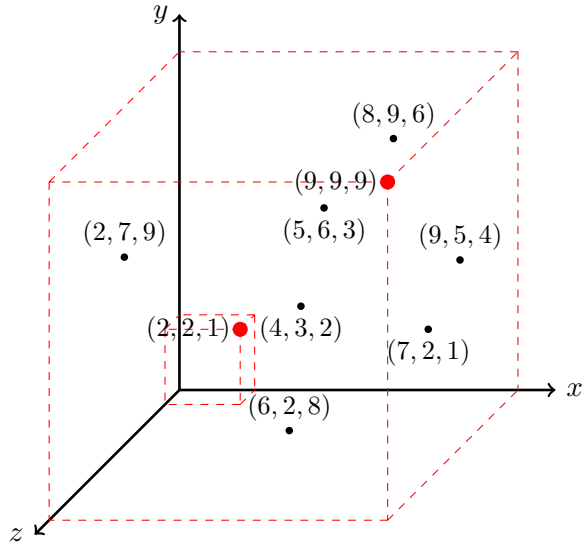
\begin{figure}[tbp]
\centering
\begin{tikzpicture} [scale=0.5]
\draw [->,line width=1pt] (0,0,0)--(10,0,0) node[right]{$x$};
\draw [->,line width=1pt] (0,0,0)--(0,10,0) node[left]{$y$};
\draw [->,line width=1pt] (0,0,0)--(0,0,10)node[left]{$z$};
\draw [dashed,red]
    (9,9,9)--(0,9,9) (0,9,9)--(0,0,9) (0,9,9)--(0,9,0)
    (9,9,9)--(9,0,9) (9,0,9)--(9,0,0) (9,0,9)--(0,0,9)
    (9,9,9)--(9,9,0) (9,9,0)--(9,0,0) (9,9,0)--(0,9,0);
\draw [dashed,red]
    (2,2,1)--(0,2,1) (0,2,1)--(0,0,1) (0,2,1)--(0,2,0)
    (2,2,1)--(2,0,1) (2,0,1)--(2,0,0) (2,0,1)--(0,0,1)
    (2,2,1)--(2,2,0) (2,2,0)--(2,0,0) (2,2,0)--(0,2,0);
\foreach \position in {(4,3,2),(7,2,1),(5,6,3)}%
\fill \position circle (0.1cm)node[below,black]{\small$\position$};
\foreach \position in {(9,5,4),(8,9,6),(2,7,9),(6,2,8)}%
\fill \position circle (0.1cm) node[above]{\small$\position$};
\fill[red] (9,9,9) circle (0.2cm) node[left,black]{\small$(9,9,9)$};
\fill[red] (2,2,1) circle (0.2cm) node[left,black]{\small$(2,2,1)$};
\end{tikzpicture}
\caption{Consider the subtree containing the points $
\{(4,3,2),(9,5,4),(7,2,1),(5,6,3),(8,9,6),(2,7,9),(6,2,8)\}$. Then
$(9,9,9)$ and $(2,2,1)$ are the upper bound and the lower bound of
the subtree, respectively.} \label{bound}
\end{figure}

Note that although additional comparisons and spaces are needed for
implementing the bounding boxes in maxima-finding algorithms, the
overall performance is generally improved, especially, when dealing
with samples with a large number of maxima.

\subsection{The proposed algorithm}

We give in this subsection our two-phase maxima-finding algorithm
using $k$-d trees and bounding boxes. In this algorithm, we need
only the upper bounds of the bounding boxes since in each phase we
only detect if the new-coming element is dominated by existing
records. An illustrative example is given in Figure~\ref{kdbound}.

For implementation details, the records are stored, during the first
phase, not only in a $k$-d tree but also in a list to preserve the
order of the records.

\begin{figure}[tbp]
\begin{minipage}{0.5\textwidth}
\centering
\begin{tikzpicture} [line width=1pt,level distance=16mm]
\tikzstyle{level 1}=[sibling distance=36mm]%
\tikzstyle{level 2}=[sibling distance=24mm] \node{ $\v{p}_1$}
  child {
    node{ $\v{p}_3$ }
    child {node{ $\v{p}_6$ }}
    child {node{ $\v{p}_4$ }}
  }
  child {
    node{ $\v{p}_2$}
    child {node{ $\v{p}_5$}}
    child [fill=none] {edge from parent[draw=none]}
    };
\end{tikzpicture}
\end{minipage}
\begin{minipage}{0.5\textwidth}
\centering
\begin{tikzpicture} [scale=0.5]
\fill (10,6) circle (0.1cm) node[above right,black]
{$\v{u}_1$};%
\fill (0,6) circle (0.1cm) node[left ,black] {$\v{p}_4$};%
\fill (1,5) circle (0.1cm) node[below,black] {$\v{p}_3$};%
\fill (3,6) circle (0.1cm) node[above right,black] {$\v{u}_3$};
\fill (3,4.5) circle (0.1cm) node[below,black] {$\v{p}_6$};%
\fill (4,4) circle (0.1cm) node[below,black] {$\v{q}$};%
\fill (6,3) circle (0.1cm) node[below,black] {$\v{p}_1$};%
\fill (8,2) circle (0.1cm) node[below,black] {$\v{p}_2$};%
\fill (10,2) circle (0.1cm) node[right,black] {$\v{u}_2$};%
\fill (10,0.5) circle (0.1cm) node[below,black] {$\v{p}_5$};%
\draw [dashed,line width=0.2pt] (10,6)--(0,6);%
\draw [dashed,line width=0.2pt] (3,6)--(3,4.5);%
\draw [dashed,line width=0.2pt] (8,2)--(10,2);%
\draw [dashed,line width=0.2pt] (10,0.5)--(10,6) ;
\end{tikzpicture}
\end{minipage}
\caption{Consider the $k$-d tree with six points
$\v{p}_1$,$\v{p}_2$,\dots, $\v{p}_6$ and a new point $\v{q}$. The
upper bounds of the trees rooted at $\v{p}_1$, $\v{p}_2$ and
$\v{p}_3 $ are $\v{u}_1$, $\v{u}_2$ and $\v{u}_3$, respectively. To
check if $\v{q}$ is dominated by some point in the tree, the
comparisons between $\v{q}$ and subtrees rooted at $\v{p}_2$ and
$\v{p}_3$ can all be skipped since $\v{q}$ is not dominated by
$\v{u}_2$ and $\v{u}_3$. } \label{kdbound}
\end{figure}
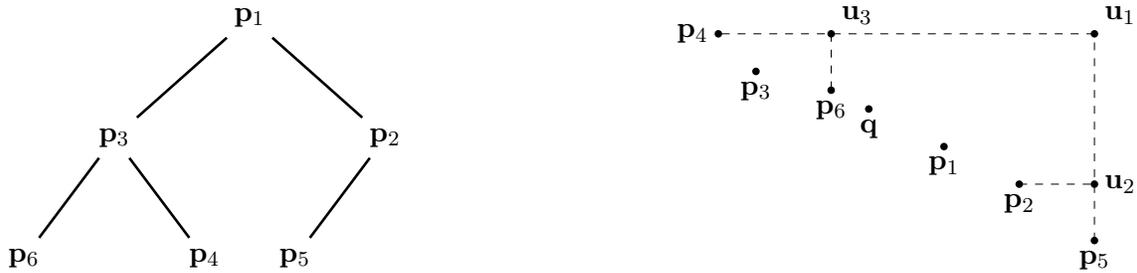

\medskip

\noindent \hspace*{1cm}\textbf{Algorithm} \textsf{Maxima}\\
\hspace*{1cm}//\textbf{Input}: A sequence of points
$\v{p}=\{\v{p}_1,\dots , \v{p}_{n}\}$\\
\hspace*{1cm}//\textbf{Output}: a $k$-d tree rooted at
$\v{r}$ consisting of $ \v{Max}(\v{p})$\\%
\hspace*{1cm}\textbf{begin}\\%
\hspace*{2cm} $\v{r}:=\v{p}_1;
\v{u}_{\v{r}}:=\v{p}_1$\\%
\hspace*{2cm} $\v{q}_1:=\v{p}_1$ \qquad
// $\v{R}:=\{ \v{q}_1\},$ the sequence of the records.\\
\hspace*{2cm} $k:=1$ \qquad // $k$ counts the number of records \\
\hspace*{2cm} \textbf{for} $i:=2$ \textbf{to} $n$ \textbf{do} \\%
\hspace*{3cm} \textbf{if} (\textsf{Dominated}$(\v{r},
\v{p} _i)=0 $) \textbf{then} \\%
\hspace*{4cm}\textsf{Insert}$(\v{r},1,\v{p}_i)$;\\%
\hspace*{4cm} $k:=k+1$; $\v{q}_k:=\v{p}_i$\qquad
\\
\hspace*{2cm}// $\v{R}=\{\v{q}_1,\dots ,\v{q}_k\}$
when $ i=n$\\%
\hspace*{2cm} release the tree rooted at $\v{r}$\\%
\hspace*{2cm} $\v{r}:=\v{q}_k;
\v{u}_{\v{r}}:=\v{q}_k$;\\%
\hspace*{2cm} \textbf{for} $i:=k-1$ \textbf{downto} $1$
\textbf{do}\\ \hspace*{3cm} \textbf{if} (\textsf{Dominated}$(\v{r},
\v{q} _i)=0$) \textbf{then}
\textsf{Insert}$(\v{r},1,\v{q}_i)$ \\
\hspace*{1cm}\textbf{end}

\quad

\noindent
\hspace*{1cm}\textsf{Dominated}$(\v{r},\v{p})$\\
\hspace*{1cm}//\textbf{Input}: A node $\v{r}$ in a $k$-d tree
and a point $\v{p}$\\%
\hspace*{1cm}//\textbf{Output}: $\left\{\begin{array}{ll}
    0, & \text{if $\v{p}$ is not dominated
    by any point in the subtree rooted at $\v{r}$}\\
    1, & \text{otherwise}
    \end{array}\right. $\\ %
\hspace*{1cm}\textbf{begin}\\ \hspace*{2cm} \textbf{if}
($\v{p}\prec \v{r}$) \textbf{then return} 1\\%
\hspace*{2cm} \textbf{if} ($\v{r}.\text{left}\ne \emptyset $ and $
\v{p}\prec \v{u}_{\v{r}.\text{left}}$)
\textbf{then} \\
\hspace*{3cm} \textbf{if} (\textsf{Dominated}($\v{r}.\text{left},
\v{p})=1$)
\textbf{then return} 1\\
\hspace*{2cm} \textbf{if} ($\v{r}.\text{right}\ne \emptyset $ and $
\v{p}\prec
\v{u}_{\v{r}.\text{right}}$) \textbf{then} \\
\hspace*{3cm} \textbf{if} (\textsf{Dominated}($\v{r}.\text{right},
\v{p})=1$)
\textbf{then return} 1 \\
\hspace*{2cm} \textbf{return} 0\\
\hspace*{1cm} \textbf{end}

\quad

\noindent \hspace*{1cm}\textsf{Insert}$(\v{r},\ell ,\v{p})$\\
\hspace*{1cm}\textbf{begin} \\ \hspace*{2cm}
$\v{u}_{\v{r}}:=\text{max} \{\v{u}_{\v{r}}, \v{p}\}$ \quad
// update the upper bound\\ \hspace*{2cm} compare the
$\ell$-th component of $\v{p}$ and that of $ \v{r}$\\
\hspace*{2cm} \textbf{Case 1:} $\v{p}_{\ell }\ge r_{\ell
}$ and $\v{r}.{ \text{right}}\ne \emptyset $\\
\hspace*{3cm} \textsf{Insert}$(\v{r}.{\text{right}},1+\ell\
\text{mod }\ d,\v{p})$\\ \hspace*{2cm} \textbf{Case 2:} $\v{p}_{\ell
}\ge r_{\ell }$ and $\v{r}.{ \text{right}}=\emptyset $\\
\hspace*{3cm} $\v{r}.{\text{right}}:=\v{p}$;
$\v{u}_{\v{r}.{\text{right}}}:=\v{p}$\\
\hspace*{2cm} \textbf{Case 3:} $\v{p}_{\ell }<r_{\ell }$ and
$\v{r}.{\text{ left}}\ne \emptyset $\\ \hspace*{3cm}
\textsf{Insert}$(\v{r}.{\text{left}},1+\ell\ \text{mod} \
d,\v{p})$\\ \hspace*{2cm} \textbf{Case 4:} $\v{p}_{\ell
}<r_{\ell }$ and $\v{r}.{\text{ left}}=\emptyset $\\
\hspace*{3cm} $\v{r}.{\text{left}}:=\v{p}$;
$\v{u}_{\v{r} .{\text{left}}}:=\v{p}$\\
\hspace*{1cm}\textbf{end}

\quad

Note that the upper bound of a subtree is updated after a new point
is inserted. In the procedure \textsf{Dominated}, the ``filtering
role" played by the upper bounds may quickly reduce many
comparisons. In practice, if a point $\v{p}$ is not dominated by
$\v{u}_{\v{r}.{\text{left}}}$ (or $\v{u}_{ \v{r}.{\text{right}}}$),
then $\v{p}$ is not dominated by any point in the subtree and the
comparisons between $\v{p}$ and the points of the subtree are all
skipped.

\subsection{Further improvements: sieving and pruning}
\label{sec-fi}

The algorithm \textsf{Maxima} is not on-line in nature since it
requires two passes through the input. In this section, we discuss
sieving and periodic pruning techniques, and present
an on-line algorithm.

\paragraph{Sieving} The idea is to select an element (or several
elements) as a good sieve (or ``keeper"), so as to dominate as many
as possible in-coming points, thus reducing the total number of
comparisons made. This was first introduced in \cite{BCL93}.

For our algorithm \textsf{Maxima}, many of the points inserted
into the $k$-d tree may have limited power of dominating in-coming
points. We can improve further Algorithm \textsf{Maxima} by choosing
the input point with the largest $L^1$-norm (which is the sum of the
absolute values of all coordinates) to be the sieve and incorporate
such a procedure as part of algorithm \textsf{Maxima}. The resulting
implementation is very efficient, notably for samples with only a
small number of maxima.

A simple way to incorporate the maximum $L^1$-norm point is to
replace the line
\begin{quote}
\textbf{for} $i:=2$ \textbf{to} $n$ \textbf{do}
\end{quote}
in algorithm \textsf{Maxima} by the following
\begin{quote}
\hspace*{2cm} $\v{s}:=\v{p}_1$ \qquad //
$\v{s}=$ sieve \\%
\hspace*{2cm} \textbf{for} $i:=2$ \textbf{to} $n$ \textbf{do}\\%
\hspace*{3cm} \textbf{if} ($\v{p}_i\nprec \v{s}$)
\textbf{then} \\%
\hspace*{4cm} $\v{s} :=\left\{\begin{array}{ll}
    \v{s},&\textbf{if }
    \norm{\v{s}} \ge \norm{\v{p}_i};\\
    \v{p}_i,& \textbf{if }
    \norm{\v{s}}<\norm{\v{p}_i},
\end{array}\right.$
\end{quote}
where $\norm{\cdot}$ denotes the $L^1$-norm. Thus the sieving
process is carried out only during the first phase. Other sieves can
be considered similarly.

\paragraph{Pruning} In the first phase of Algorithm
\textsf{Maxima}, the $k$-d tree may contain some nodes that are
dominated by other nodes in the tree, and will only be removed in
the second phase of the algorithm. In particular, if the dominated
nodes are close to the root, then more comparisons may be made. It
is thus more efficient to carry out an initial pruning of the $k$-d
tree by removing dominated points in the tree after a sufficiently
large number of records have been inserted (and still small compared
with the total sample size). Such an early pruning idea can be
implemented by running the following procedure.

\medskip

\noindent \hspace*{1cm}\textbf{Algorithm} \textsf{Prune}\\ %
\hspace*{1cm}// only called once in the first
\textbf{for}-loop of Algorithm
\textsf{ \textsf{Maxima}}\\ %
\hspace*{1cm}// Assume $\v{R}=\{\v{q}_{{1}},\ldots
,\v{q}_{{K} }\}$\\%
\hspace*{1cm}\textbf{begin}\\%
\hspace*{2cm} release the $k$-d tree\\%
\hspace*{2cm} $\v{r} :=\v{q}_{{K}};
\v{u}_{\v{r}}:=\v{ q}_{{K}}$\\%
\hspace*{2cm} \textbf{for} $j:=K-1$ \textbf{downto} $1$ \\
\hspace*{3cm} \textbf{if} (\textsf{Dominated}$(\v{r},\v{q} _j)=0 $)
\textbf{then}
\textsf{Insert}$(\v{r},1,\v{q}_j)$\\
\hspace*{1cm}\textbf{end}\\

We can call \textsf{Prune} when, say $i=\lfloor n/\lambda
\rfloor $ or $i=\lfloor n^{\delta }\rfloor $,
where $i$ is the index in the first \textbf{for}-loop of algorithm
\textsf{Maxima}.
For example, we can take $\lambda =10$ and $\delta
=2/3$. Which choice is optimal is an interesting issue but depends
on the practical implementations. Also one may consider the use
of periodic pruning, but since pruning is a costly operation, we
chose to apply it only once in our simulations.

\paragraph{An on-line algorithm} On-line maxima-finding algorithms
always retain the maxima of the all input points read so far and are
often needed in many practical situations. A simple means to convert
our algorithm \textsf{\textsf{Maxima}} into an on-line one is to add
a procedure to delete the dominated elements in the $k$-d tree. The
deletions can be made immediately after comparison with each
in-coming element, which results in restructuring the whole $k$-d
tree and may be very costly if the elements deleted are not near the
bottom of a large tree. A simple way to perform the deletion of a
node is to re-insert all its descendant nodes one by one, in the
order inherited from the original input sequence. However, the
procedure can be time-consuming and the resulting tree may be quite
imbalanced.

We introduce an on-line implementation by storing the current maxima
in an extra list. In each iteration, we look for all points in the
$k$-d tree that are dominated by the in-coming point $\v{p}$, mark
them, and delete the corresponding elements from the extra list. The
lower bounds of the bounding boxes are useful here. Recall
${\v{v}}_{\v{r}}=(v_1,\dots ,v_d)$, where $v_i$ is the minimum among
all the $i$-th coordinates of points in the subtree rooted at
$\v{r}$. When searching for those points in $\v{M}$ that are
dominated by $\v{p}$, we can skip checking the subtree of $\v{r}$ if
$\v{v}_{\v{r}}$ is not dominated by $\v{p}$.

The on-line algorithm is given as follows.

\quad

\noindent \hspace*{1cm}\textbf{Algorithm}
\textsf{On-Line-Maxima}\\%
\hspace*{1cm}//\textbf{Input}: A sequence of points
$\v{p}=\{\v{p}_1,\dots , \v{p}_{n}\}$\\%
\hspace*{1cm}//\textbf{Output}: $\v{M}:=$ the list containing
$\v{Max}(\v{p})$\\%
\hspace*{1cm}\textbf{begin}\\%
\hspace*{2cm}$\v{r}:=\v{p}_1; \v{u}_{\v{r}}:=\v{p} _1; \v{v}_{\v{r}}
:=\v{p}_1$ \\%
\hspace*{2cm} $\v{M}:=\{\v{p}_1\}$\\%
\hspace*{2cm} \textbf{for} $i:=2$ \textbf{to} $n$
\textbf{do}\\%
\hspace*{3cm} \textbf{if} (\textsf{Dominated}$(\v{r},\v{p}
_i)=0 $) \textbf{then} \\%
\hspace*{4cm} \textsf{Delete}$(\v{r},\v{p}_i)$\\
\hspace*{4cm} \textsf{Insert}$(\v{r},1,\v{p}_i)$ \\
\hspace*{4cm} $\v{M}:=\v{M}\cup \{\v{p}_i\}$\\
\hspace*{1cm}\textbf{end}

\quad

\noindent \hspace*{1cm}\textsf{Delete}$(\v{r},\v{p})$
\\ \hspace*{1cm}//\textbf{Input}: A node $\v{r}$ of a
$k$-d tree and a point $\v{p}$\\
\hspace*{1cm}//\textbf{Output}: a more compact $\v{M}$ (all
dominated points are removed) \\
\hspace*{1cm}\textbf{begin}\\ \hspace*{2cm} \textbf{if} ($\v{r}\prec
\v{p}$) \textbf{then}
\\
\hspace*{3cm} \textbf{if} ($\v{r}$ is unmarked) \textbf{then\qquad
}// The set of unmarked nodes is exactly $\v{M}$
\\ \hspace*{4cm} delete $\v{r}$ from $\v{M}$
\\ \hspace*{4cm} mark $\v{r}$\\ \hspace*{2cm}
\textbf{if} ($\v{r}.\text{left}\ne \emptyset $ and $
\v{v}_{\v{r}.\text{left}}\prec \v{p}$)
\textbf{then} \textsf{Delete}($\v{r}.\text{left},\v{p}$)\\
\hspace*{2cm} \textbf{if} ($\v{r}.\text{right}\ne \emptyset $ and $
\v{v}_{\v{r}.\text{right}}\prec \v{p}$)
\textbf{then} \textsf{Delete}($\v{r}.\text{right},\v{p}$)\\
\hspace*{1cm}\textbf{end}

\quad

Note that the only difference between the procedure \textsf{Insert}
of algorithm \textsf{On-Line-Maxima} and that of algorithm
\textsf{Maxima} is that we need to update both the upper bounds and
the lower bounds in the procedure \textsf{Insert}$(\v{r}, j,\v{p})$
of algorithm \textsf{On-Line-Maxima}.

\subsection{Comparative discussions}

We ran a few sequential algorithms and tested their performance
under several types of random data, each with $1000$ iterations; the
average values of the results are given in Tables \ref{3cube} and
\ref{3sim}. The points are generated uniformly and independently at
random from a given region $D$, which is either a hypercube or a
simplex.

\begin{itemize}
\item list: a sequential algorithm using a linked list
(see \cite{BCL93});

\item $d$-tree: a sequential algorithm using the $d$-ary
tree proposed in \cite{Schutze03};

\item quadtree: a sequential algorithm using quadtree
(see \cite {Habenicht83,SS96,MTT02});

\item $2$-phase: algorithm \textsf{Maxima};

\item +prune: algorithm \textsf{Maxima} with an early pruning for
$i=n/10$;

\item +sieve: algorithm \textsf{Maxima} with the max-$L^1$-norm
sieve;

\item +prune\&sieve: algorithm \textsf{Maxima} with pruning for
$i=n/10$ and the max-$L^1$-norm sieve.

\end{itemize}

\begin{table}[!htbp]
\caption{The average numbers of scalar comparisons per input point
when $D=[0,1]^d$, where $d\in\{3,4,6,10\}$.} \label{3cube}\centering
{\small \centering {\renewcommand{\arraystretch}{1.2}
\tabcolsep=2pt}}
\par{\small $d=3$ }\par
{\small
\begin{tabular}{c|ccc|cccc}
\toprule $n$ & list & d-tree & quadtree & 2-phase & +prune & +sieve
& +prune\&sieve\\\hline $10^2$ & $11.40$ & $19.38$ & $13.58$ &
$24.72$
& $23.23$ & $19.10$ & $18.82$ \\
$10^3$ & $11.01$ & $15.01$ & $11.38$ & $24.29$
& $20.81$ & $13.23$ & $12.43$ \\
$10^4$ & $8.28$ & $12.02$ & $9.41$ & $23.30$
& $17.70$ & $8.44$ & $7.69$ \\
$10^5$ & $6.36$ & $11.21$ & $8.50$ & $23.31$
& $15.70$ & $5.78$ & $5.30$ \\
$10^6$ & $5.01$ & $11.40$ & $8.07$ & $23.05$
& $13.75$ & $4.40$ & $4.09$ \\
$10^7$ & $4.24$ & $11.51$ & $7.91$ & $23.76$
& $12.50$ & $3.73$ & $3.54$\\
$10^8$ & $3.88$ & $12.02$ & $7.67$ & $24.11$
& $11.39$ & $3.36$ & $3.25$\\
\toprule
\end{tabular}
}
\par
\vspace*{0.3cm}
{\small $d=4$ }
\par
{\small
\begin{tabular}{c|ccc|cccc}
\toprule $n$ & list & d-tree & quadtree & 2-phase & +prune & +sieve
& +prune\&sieve\\ \hline
$10^2$ & $26.96$ & $47.28$ & $30.29$ & $50.22$
& $50.05$ & $44.28$ & $44.78$ \\
$10^3$ & $37.41$ & $49.48$ & $31.53$ & $53.28$
& $51.07$ & $38.43$ & $37.76$ \\
$10^4$ & $32.48$ & $40.62$ & $26.80$ & $48.34$
& $43.94$ & $25.73$ & $24.79$ \\
$10^5$ & $22.36$ & $34.32$ & $22.60$ & $44.30$
& $37.75$ & $16.65$ & $15.78$ \\
$10^6$ & $14.69$ & $32.36$ & $20.66$ & $42.69$
& $33.00$ & $11.32$ & $10.61$ \\
$10^7$ & $10.08$ & $32.46$ & $19.47$ & $42.74$
& $29.87$ & $8.40$ & $7.80$\\
$10^8$ & $8.40$ & $33.04$ & $19.05$ & $52.22$
& $28.88$ & $6.83$ & $6.08$\\
\toprule
\end{tabular}
}
\par
\vspace*{0.3cm}
{\small $d=6$ }
\par
{\small
\begin{tabular}{c|ccc|cccc}
\toprule $n$ & list & d-tree & quadtree & 2-phase & +prune & +sieve
& +prune\&sieve\\ \hline $10^2$ & $75.44$ & $139.19$ & $74.32$
&$129.85$ & $131.41$
& $126.53$ & $128.20$\\
$10^3$ & $228.69$ & $284.69$ & $130.37$ & $193.84$
& $193.27$ & $177.23$ & $177.44$\\
$10^4$ & $384.86$ & $343.69$ & $149.75$ & $194.56$
& $194.05$ & $163.10$ & $163.17$\\
$10^5$ & $404.74$ & $298.21$ & $131.41$ & $162.01$
& $161.27$ & $116.86$ & $117.40$\\
$10^6$ & $310.75$ & $222.30$ & $104.53$ & $133.34$
& $131.66$ & $77.55$ & $78.68$\\
$10^7$ & $190.08$ & $166.02$ & $86.63$ & $118.09$
& $112.34$ & $52.13$ & $52.65$\\
$10^8$ & $100.77$ & $136.69$ & $74.97$ & $109.50$
& $98.93$ & $36.46$ & $36.36$\\
\toprule
\end{tabular}
}
\par
\vspace*{0.3cm}
{\small $d=10$ }
\par
{\small
\begin{tabular}{c|ccc|cccc}
\toprule $n$ & list & d-tree & quadtree & 2-phase & +prune & +sieve
& +prune\&sieve\\ \hline
$10^2$ & $137.56$ &$296.70$&$132.72$& $267.90$
& $270.67$ & $269.49$ & $272.22$\\
$10^3$ & $1048.73$ &$1496.07$&$458.30$& $774.85$
& $777.16$ & $769.01$ & $771.42$\\
$10^4$ & $5392.57$ &$4916.40$&$1190.22$& $1526.83$
& $1528.66$ & $1498.47$ & $1499.93$\\
$10^5$ & $17779.34$ &$11463.01$&$2201.99$& $2126.49$
& $2132.18$ & $2062.42$ & $2067.98$\\
$10^6$ & $38552.96$ &$18775.90$ &---& $2221.26$
& $2234.51$ & $2121.11$ & $2132.94$\\
$10^7$ & $59207.23$ &$20769.36$ &---& $2023.64$
& $1844.68$ & $1931.37$ & $1750.01$\\
$10^8$ & --- &$19226.26$ &---& $1544.68$ & $1387.00$
& $1429.45$ & $1261.90$\\
\toprule
\end{tabular}
}
\end{table}

\begin{table}[!htbp]
\caption{The average numbers of scalar comparisons per input point
when $D$ is the $d$-dimensional simplex, where $d=3,4$ and $6$.}
\label{3sim}\centering {\small \centering
{\renewcommand{\arraystretch}{1.2} \tabcolsep=2pt }}
\par
\vspace*{0.3cm}
{\small $d=3$ }
\par
{\small
\begin{tabular}{c|ccc|cccc}
\toprule
$n$ & list & d-tree & quadtree & 2-phase
& +prune & +sieve & +prune\&sieve\\
\toprule
$10^2$ & $40.96$ & $62.81$ & $30.50$ & $57.68$
& $58.00$ & $57.87$ & $58.26$\\
$10^3$ & $134.05$ & $112.71$ & $43.98$ & $82.03$
& $80.78$ & $81.34$ & $80.24$\\
$10^4$ & $357.25$ & $203.97$ & $55.91$ & $95.20$
& $92.37$ & $93.78$ & $91.23$\\
$10^5$ & $858.65$ & $402.18$ & $76.19$ & $105.64$
& $100.79$ & $104.10$ & $99.59$\\
$10^6$ & $1957.22$ & $835.16$ & $126.45$ & $117.42$
& $107.53$ & $117.11$ & $107.60$\\
$10^7$ & $4334.09$ & $1678.73$ & $161.25$ & $129.18$
& $106.81$ & $130.72$ & $108.50$\\
$10^8$ & $9417.80$ & $3543.73$ & $331.25$ & $142.22$
& $116.74$ & $142.73$ & $116.98$\\
\toprule
\end{tabular}
}
\par
\vspace*{0.3cm}
{\small $d=4$ }
\par
{\small
\begin{tabular}{c|ccc|cccc}
\toprule $n$ & list & d-tree & quadtree & 2-phase & +prune & +sieve
& +prune\&sieve\\ \hline
$10^2$ & $81.74$ & $123.95$ & $57.61$ & $107.18$
& $108.36$ & $108.37$ & $109.55$ \\
$10^3$ & $441.09$ & $368.00$ & $117.09$
& $199.20$ & $199.35$ & $199.70$ & $199.87$ \\
$10^4$ & $1917.26$ & $910.44$ & $208.67$
& $287.21$ & $286.60$ & $287.09$ & $286.49$ \\
$10^5$ & $7316.79$ & $2230.39$ & $356.48$
& $373.86$ & $371.80$ & $373.60$ & $371.60$ \\
$10^6$ & $25786.00$ & $5948.65$ & $614.88$
& $474.28$ & $460.27$ & $474.84$ & $461.06$ \\
$10^7$ & $86609.63$ & $17071.62$ & $1302.10$
& $532.85$ & $487.16$ & $534.66$ & $489.15$ \\
$10^8$ & --- & $53140.49$ & $4696.73$ & $651.13$
& $698.55$ & $646.59$ & $693.70$ \\
\toprule
\end{tabular}
}
\par
\vspace*{0.3cm}
{\small $d=6$ }
\par
{\small
\begin{tabular}{c|ccc|cccc}
\toprule $n$ & list & d-tree & quadtree & 2-phase & +prune & +sieve
& +prune\&sieve\\ \hline
$10^2$ & $126.37$ & $221.77$ & $91.42$ & $175.93$
& $177.77$ & $177.79$ & $179.63$ \\
$10^3$ & $1096.21$ & $1175.40$ & $268.67$ & $467.27$
& $468.87$ & $468.96$ & $470.56$ \\
$10^4$ & $8284.26$ & $5660.90$ & $758.05$ & $993.25$
& $995.64$ & $994.77$ & $997.16$ \\
$10^5$ & $55200.49$ & $24332.05$ & $2178.38$ & $1849.37$
& $1856.49$ & $1850.86$ & $1858.01$ \\
$10^6$ & $331776.01$ & $93275.52$ & $6825.69$ & $3153.92$
& $3125.31$ & $3155.81$ & $3127.10$ \\
$10^7$ & --- & $368306.29$ & $8418.26$ & $5090.63$
& $5029.78$ & $5092.54$ & $5031.71$ \\
$10^8$ & --- & --- & --- & $7996.92$ & $7403.24$
& $7998.93$ & $7405.39$ \\
\toprule
\end{tabular}
}
\end{table}

Table \ref{3cube} shows evidently that our two-phase maxima-finding
algorithms, whether coupling with sieving and pruning techniques,
perform very well under random inputs from the $d$-dimensional
hypercubes. They are efficient and uniformly scalable since the
average number of scalar comparisons each point involved is
gradually rising, in contrast to the fast increase of other
algorithms. Note that, according to a result by Devroye
\cite{Devroye99}, we expect that the average number of scalar
comparisons each point involves tends eventually to $d$ in each
case. This is visible for $d=3$ but less clear for higher values of
$d$, as the convergence rate is very slow. Also the numbers in each
column first increases as $n$ increases and then decreases.

On the other hand, although the asymptotic growth rate of the
expected numbers of maxima $\mu_{n,d}$ in such cases are
approximately $(\log n)^{d-1}/(d-1)!$ for large $n$ and fixed $d$,
the real values of $\mu_{n,d}$ for moderate $d$ soon become large;
for example, when $d=10$
\[
    \{\mu_{10^i,10}\}_{i=2,\dots,8}
    \approx\{94, 765, 4\,947, 25\,113, 103\,300, 357\,604,
    1\,076\,503\}.
\]
These values were computed by the recurrence (see
\cite{BDHT05})
\[
    \mu_{n,d} = \frac1{d-1} \sum_{1\le j<d}
    H_n^{(d-j)}\mu_{n,j}\qquad(d\ge2),
\]
with $\mu_{n,1} = 1$ for $n\ge1$, where the $H_n^{(j)} := \sum_{1\le
i\le n} 1/i^j$ are Harmonic numbers. They can also be estimated
by the asymptotic approximations given in \cite{BDHT05}.

The situation is very similar (see Table~\ref{3sim}) when the random
samples are generated from the $d$-dimensional simplex,
$D=\{\v{x}:x_i\ge 0,\sum_{1\le i\le d} x_i\le 1\}$ for which the
expected numbers of maxima $\nu_{n,d}$ are of order $n^{1-1/d}$
instead of $(\log n)^{d-1}$; see \cite{BDHT05}. In such cases,
$\nu_{n,d}$ grows even faster than $\mu_{n,d}$. For example, when
$d=6$,
\[
    \{\mu_{10^i,6}\}_{i=2,\dots,8}
    \approx\{95, 863, 7\,281, 57\,858, 439\,110, 3\,223\,774,
    23\,121\,832\}.
\]
These values were computed by the exact formula
\[
    \nu_{n,d} = n\sum_{0\le j<d} \binom{d-1}{j}
    (-1)^j\frac{\Gamma(n)\Gamma((j+1)/d)}{\Gamma(n+(j+1)/d)}
    \qquad(d\ge2),
\]
which follows from
\begin{align*}
    \nu_{n,d} &= n\mb{P}(\v{x}_1 \text{ is a maxima})\\
    &= d! n \int_D\left(1-\left(1-
    {\textstyle\sum}_{1\le i\le d} x_i\right)^d\right)^{n-1}
    \text{d} \v{x} \\
    &= dn \int_0^1 \left(1-(1-y)^d\right)^{n-1} y^{d-1} \text{d} y,
\end{align*}
by straightforward calculations, where $\Gamma$ denotes the Gamma
function. For similar details, see \cite{BDHT05}.

Unlike hypercubes where sieving is seen to be very helpful, the gain
of sieving for random samples whose coordinates are roughly
negatively correlated is marginal since there is no ``omnipotently
powerful" point; see \cite{Baryshnikov07,Golin93}.

A feature of the quadtree algorithm is that by its large amount of
branching factors ($2^d-2$), the position of a point in the tree is
quickly identified, often after a few comparisons, and the bounding
boxes are thus not helpful here. We also tested $2$-phase quadtree
and $2$-phase $d$-tree algorithms, the improvement over the original
algorithms is much more significant in $d$-trees than in quadtrees.
In contrast, since $k$-d trees are binary, the use of the bounding
boxes plays a crucial role in accelerating the performance of the
algorithm.

Note that the data collected in these two tables do not reflect
directly the running time of each program. In terms of running
time, our algorithms perform much better than the others.

Simulations also suggested that our on-line algorithm is also
reasonably efficient when compared with other algorithms.

\section{Average-case analysis of algorithm $\textsf{Maxima}$}
\label{sec-analysis}

We derive in this section a few analytic results in connection with
the performance of the algorithms we proposed in this paper. In
general, probabilistic analysis of sequential algorithms for finding
the maxima of random samples is very difficult due to the dynamic
nature of the algorithms; see \cite{Devroye99, Golin94} and the
references therein.

\subsection{How many non-dominated records are there?}

The performance of \textsf{Maxima} depends heavily on the number of
records, which in turn is closely related to the number of maxima.

\begin{thm} Let $R_{n}$ denote the number of non-dominated records
in a sequence $\{\v{p}_1,\dots,\v{p}_n\}$ of independent and
uniformly distributed points from some region $D$ in $\mb{R}^d$. Let
$M_n$ denote the maxima of $\{\v{p}_1,\dots, \v{p}_n\}$. Then
\begin{align} \label{Rn-Mn}
    \mb{E}(R_n)=\sum_{i=1}^n\frac{\mb{E}(M_i)}{i}.
\end{align}
\end{thm}
\begin{proof}
By assumption,
\[
    \mb{P}\left( \v{p}_n\in \v{Max}
    (\{\v{p}_1,\dots ,\v{p}_n\})\right)
    =\cdots
    =\mb{P}\left( \v{p}_n\in
    \v{Max}(\{\v{p}_1,\dots ,\v{p}_n\})\right).
\]
Thus
\[
    \mb{E}(M_n)
    =\sum_{i=1}^n\mb{P}\left(
    \v{p}_i\in\v{Max}(\{\v{p}_1,
    \dots ,\v{p}_n\})\right)
    =n\mb{P}\left( \v{p}_n\in\v{Max}
    (\{\v{p}_1,\dots ,\v{p}_n\})\right) .
\]
Then we have
\begin{align*}
    \mb{E}(R_n)
    &=\sum_{i=1}^n\mb{E}\v{1}_{
    \left(\v{p}_i\text{ is a record}\right) } \\
    &=\sum_{i=1}^n\mb{P}\left( \v{p}_i\in
    \v{Max}(\{\v{p}_1,\dots ,\v{p}_i\})\right)  \\
    &=\sum_{i=1}^n\frac{\mb{E}(M_i)}{i}.
\end{align*}
\end{proof}

Since $\mb{E}(M_n)$ is usually of order $n^\alpha$ or $(\log
n)^\beta$ for some $\alpha ,\beta \ge 0$ (see
\cite{BDHT05,BHLT01,Devroye93}), if we assume that $\mb{E}(M_n) \sim
c n^\alpha(\log n)^\beta$, where $c, \beta>0$ and $\alpha\in[0,1]$,
then, by (\ref{Rn-Mn}),
\[
    \mb{E} (R_n) \sim
    \left\{\begin{array}{ll}
        \displaystyle \frac{c}{\alpha}n^{\alpha}(\log n)^\beta
        \sim \frac{\mb{E}(M_n)}{\alpha},
        &\text{if } 0<\alpha\le 1; \\
        \displaystyle \frac{c}{\beta +1}(\log n)^{\beta +1}
        \sim \frac{\mb{E}(M_n)}{\beta+1}\log n,
        &\text{if } \alpha=0,
    \end{array}\right.
\]
where $a_n \sim b_n$ means that $a_n/b_n\to 1$ as $n\to\infty$.

In the special case when the region $D$ is the $d$-dimensional
hypercube $[0,1]^d$, then it is also easily seen that the number of
non-dominated records in random samples from $[0,1]^d$ is
identically distributed as the number of maxima in random samples
from $[0,1]^{d+1}$; see \cite{Gnedin07}.

Whichever the case, we always have
\[
    \mb{E}(R_n) \le  \mb{E}(M_n)\sum_{i=1}^n \frac1i
    = O(\mb{E}(M_n) \log n).
\]
This partly explains why our two-phase algorithm does not use much
more comparisons and runs reasonably efficient. Also we see that the
expected additional memory used for the $k$-d tree (and possibly the
array) is at most a $\log n$ factor more than the expected number of
maxima.

\subsection{Expected cost of the sieve algorithm}

Assume that $\v{p}_1,\dots, \v{p}_n$ are sampled independently and
uniformly at random from $[0,1]^d$. Let $\v{s}_n$ be the point with
the maximum $L^1$-norm. Let $\v{1} = (\underbrace{1, \dots,
1}_{d})$.
\begin{lmm} For any $c>0$, \label{lm2}
\[
    \mb{P}\left(\norm{\v{s}_n-\v{1}}
    < (cd!)^{1/d}n^{-1/d}(\log n)^{1/d}\right)
    \ge 1-n^{-c},
\]
for sufficiently large $n$.
\end{lmm}
\begin{proof}
For $0<\varepsilon<1$
\begin{align*}
    \mb{P}\left( \norm{\v{s}_n-\v{1}}
    <\varepsilon\right)
    &=1-\mb{P}\left( \norm{\v{p}_i}\le
    d-\varepsilon,1\le i\le n\right)  \\
    &=1-\left( 1-\frac{\varepsilon^{d}}{d!}\right)^n \\
    &\ge 1-e^{-\varepsilon^d n/d!}.
\end{align*}
Taking $\varepsilon=(cd!)^{1/d}n^{-1/d}(\log n)^{1/d}$, we see that
the last expression is equal to $1-n^{-c}$. Note that
$\varepsilon<1$ if $n$ is large enough. Indeed, $n/\log n>cd!$
suffices.
\end{proof}

\begin{thm} If the $n$ points $\{\v{p}_1,\dots, \v{p}_n\}$
are sampled independently and uniformly at random from $[0,1]^d$,
then the expected number of scalar comparisons used by our sieve
algorithm satisfies $dn+ O(n^{1-1/d}(\log n)^{d+1/d})$.
\end{thm}
\begin{proof}
The number of scalar comparisons used for the sieve is at most $dn$.
We claim that the expected number of the extra comparisons is only $
O(n^{1-1/d}(\log n)^{d+1/d})$. Let $a_i=(2d!)^{1/d}i^{-1/d}(\log
i)^{1/d}$. For $i$ large enough
\[
    \mb{P}\left( \norm{\v{s}_i-\v{1}}
    <a_i\right) \ge 1-i^{-2},
\]
by Lemma~\ref{lm2}. If $\v{p}_{i+1}\in [0,1-a_i]^d$ and
$\norm{\v{s}_i-\v{1}}<a_i$ both hold, then $\v{p}_{i+1} \prec
\v{s}_i$, that is, $\v{p}_{i+1}$ is filtered out. Thus, additional
comparisons are required only when either $\v{p}_{i+1}\not\in
[0,1-a_i]^d$ or $\norm{ \v{s}_i-\v{1}}\ge a_i$. If
$\v{p}_{i+1}\not\in [0,1-a_i]^d$, then the additional comparisons
used is bounded above by $O(R_{i})$; if $\norm{\v{s}_i-\v{1}} \ge
a_i$, then the extra comparisons are at most $O(i)$. Note that
$\v{p}_{i+1}$ and $R_{i}$ are independent. Thus, the expected number
of the extra comparisons required in the for-loop of $\v{p}_{i+1}$
is less than
\begin{align*}
    &\mb{P}\left( \v{p}_{i+1}\notin [0,1-a_i]^d\right)
    O(\mb{E}(R_i))+\mb{P}\left(
    \norm{\v{s}_i-\v{1}} \ge a_i\right) O(i) \\
    &=O(i^{-1/d}(\log i)^{d+1/d})+O(i^{-1})
\end{align*}
since $\mb{E}(R_i)=O\left( (\log i)^d\right) $. Summing over all
$i=2,\dots, n$, we obtain the required bound.
\end{proof}

\subsection{Expected performance of \textsf{Maxima} when all points
are maxima}

To further clarify the ``scalability" of \textsf{Maxima}, we
consider in this subsection the expected cost used by
\textsf{Maxima} under the extreme situation when the $d$-dimensional
input points are sampled independently and uniformly from the the
$(d-1)$-dimensional simplex $D=\{\v{x}\,:\,x_i\ge 0, \sum_{1\le i<d}
x_i=1\}$. Note that in the skyline context, an anti-correlated
sample is often discussed, which is the $(d-1)$-dimensional simplex
with a specified error range. In that case, most but not necessarily
all points are maxima. Since no deletion is involved in our
algorithm \textsf{Maxima}, the difference between random samples
from the $(d-1)$-dimensional simplex and the anti-correlated sample
is minor.

When $D$ is the $(d-1)$-dimensional simplex, all points are maxima,
and the time complexity of most algorithms such as the list
algorithm (see \cite{BCL93}) is of order $O(M_n^2) = O(n^2)$. We
show that the expected time complexity of $\textsf{Maxima}$ is
$O(n\log n)$ when $d=2$.

\begin{thm}\label{thm-ec} Assume that the $d$-dimensional points
$\{\v{p}_1, \cdots, \v{p}_n\}$ are independently and uniformly
distributed in the $(d-1)$-dimensional simplex. The expected number
of comparisons needed by algorithm \textsf{Maxima} for random
samples is bounded above by $O(n\log n)$ when $d=2$.
\end{thm}
We leave open the analysis for the case when $d\ge3$.
\begin{proof} Since all points in the sample are maxima, the
expected number of comparisons used in the first phase and that in
the second phase are the same. Thus, we focus on the first phase.

Assume that $\{\v{p}_1, \dots, \v{p}_m\}$ have been stored in a
$k$-d tree. We consider the number of comparisons that $\v{p}_{m+1}$
may involve inside the two procedures of the for-loop:
\textsf{Insert} and \textsf{Dominated}. The expected number of
comparisons used in \textsf{Insert} is of order
\[
    O(\text{the expected depth of the $k$-d tree})
    =O(\log m),
\]
since the $k$-d tree is essentially a binary search tree (see
\cite{Bentley75}).

We now estimate the expected number of comparisons used in
\textsf{Dominated}. Since at most three vector comparisons are
involved in the procedure \textsf{Dominated}, we analyze the number
of times $T_{m}$ the procedure \textsf{Dominated} is called. To
complete the proof, we show that $\mb{E}(T_m)=O(\log m)$.

Obviously, \textsf{Dominated}$(\v{r},\v{p}_{m+1})$ is called when
$\v{p}_{m+1}\prec \v{u}_{\v{r}}$. Thus, the number of times
\textsf{Dominated} is called is equal to the number of nodes $\v{r}$
such that $\v{p}_{m+1}\prec \v{u}_{\v{r}}$. Let $D_{\v{r} }\subset
D$ be the region that $\v{u}_{\v{r}}$ covers. Then the probability
of the event $\v{p}_{m+1}\prec \v{u}_{\v{r}}$ conditioning on the
$k$-d tree built from $\{\v{p}_1, \dots, \v{p}_m\}$ equals $\left|
D_{\v{r}}\right| /\left| D\right| $. Thus
\[
    \mb{E}(T_{m})=\frac{1}{\left| D\right|}
    \mb{E}\left(\sum_{\v{r}}
    \left| D_{ \v{r}}\right|\right) ,
\]
where the summation runs over all nodes and the expectation is taken
with respect to the $k$-d tree for $\{\v{p}_1, \dots, \v{p}_m\}$. To
estimate $\sum_{\v{r}}\left| D_{\v{r}}\right|$, we consider
$A_{\v{r} }\subset D$, the possible ranges induced by the nodes of
the subtrees rooted at $\v{r}$. The precise definition is as
follows. Define $A_{\v{r}}:=D$ when $\v{r}$ is the root. If
$\v{r}.\v{left}$ ($\v{r}.\v{right}$) represents the point at the
root node of the left (right) subtree of $\v{r}$, respectively, then
\begin{align*}
    \left\{\begin{array}{l}
    A_{\v{r}.\v{left}} :=A_{\v{r}}\cap [0,1]^{j-1}
    \times[0,x_{j})\times[0,1]^{d-j}, \\
    A_{\v{r}.\v{right}} :=A_{\v{r}}\cap [0,1]^{j-1}
    \times[x_{j},1]\times [0,1]^{d-j},
    \end{array}\right. \qquad(j=1,\dots,d),
\end{align*}
where $d=2$, the $j$-th coordinate is the discriminator of node
$\v{r}$ and $\v{r}=(x_{1},x_{2},\ldots ,x_{d})$.

Since the union of $A_{\v{r}}$ in the same level of the $k$-d tree
is at most $D$ and $D_{\v{r}}\subset A_{\v{r}}$ (see
Figure~\ref{d2}), we have
\[
    \mb{E}(T_m)
    \le \frac{1}{\left| D\right|}
    \mb{E}\left(\sum_{\v{r}}
    \left| A_{ \v{r}}\right|\right)
    \le \text{ the expected depth of the }
    k\text{-}d\text{ tree}
    =O(\log m).
\]
\end{proof}
Note that $A_{\v{r}}$ is determined by $\v{r}$ and its ancestors; in
contrast, $D_{\v{r}}$ is determined by $\v{r}$ and its offsprings.

\begin{figure}[!htbp]
\centering
\begin{tikzpicture} [scale=0.35]
\draw [densely dotted,line width=1.2pt] (4.2,-4.2)--(19.2,-15.7) ;
\draw [densely dotted,line width=0.6pt] (16,-6.87)--(16,-13.2) ;
\draw [densely dotted,line width=0.6pt] (16,-6.87)--(7.7,-6.87) ;
\draw [|-|,line width=2pt,blue] (4.2,-4.2)--(19.2,-15.7) ;%
\draw [|-|,line width=1.5pt,red] (7.59,-6.9)--(15.94,-13.35) ;
\filldraw (16,-6.87) circle (5pt);%
\filldraw (16,-13.2) circle (5pt);%
\filldraw (7.7,-6.87) circle (5pt);%
\filldraw (10.5,-9) circle (5pt);%
\filldraw (14,-11.68) circle (5pt);%
\draw (16.5,-5.8) node[black]{$\v{u_r}$};%
\draw (14,-10.68) node[black]{$\v{r}$};%
\draw (11,-11) node[red]{$\v{D_r}$};%
\draw (4.5,-6.5) node[blue]{$\v{A_r}$};
\end{tikzpicture}
\caption{A possible configuration of $A_{\v{r}}$ and $D_{\v{r}}$ for
$d=2$.} \label{d2}
\end{figure}
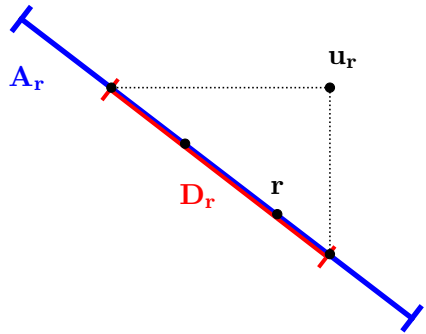

For $d\ge 3$, the expected time-complexity remains open. However,
simulations suggest that for fixed $d$ the expected time be of order
$O(n (\log n)^c)$ for some $c>0$; see Figure~\ref{effect}. On the
other hand, for fixed $n$ and increasing $d$, the expected number of
comparisons appears to be of order $O\left( dn\log n\right)$.
\begin{figure}[!ht]
\centering
\begin{tikzpicture}[scale=0.7]
\draw(-3.00,5.00)node{\normalsize
$\frac{\sum_{i=1}^{n}T_i}{n\log_2{n}}$};
\draw(5.00,-1.50)node{\normalsize$k$};%
\draw[line width=1pt,->](0,0)--(11.40,0); \foreach \x/\k in
{0.20/10,2.20/12,4.20/14,6.20/16,8.20/18,10.20/20} {%
\draw[line width=1pt](\x,-0.2)--(\x,0);
\draw(\x,-0.5)circle(0pt)node[black]{\normalsize$\k$}; }%
\foreach \x in {1.20,3.20,5.20,7.20,9.20}%
\draw[line width=1pt](\x,-0.1)--(\x,0);%
\draw[line width=1pt,->](0,0)--(0,10.80);%
\foreach \y/\k in
{0.00/0,1.11/2,2.22/4,3.33/6,4.44/8,5.56/10,6.67/12,
7.78/14,8.89/16,10.00/18} {%
\draw[line width=1pt](-0.2,\y)--(0,\y);
\draw(-1,\y)circle(0pt)node[black]{\normalsize$\k$}; }%
\foreach \y in {0.56,1.67,2.78,3.33,3.89,5.00,6.11,7.22,8.33,9.44}
\draw[line width=1pt](-0.1,\y)--(0,\y);%
\draw[line width=1pt,color=red,mark=o,mark size=3pt,draw=red] plot
coordinates{(0.20,0.81)(1.20,1.08)(2.20,1.32)(3.20,1.76)
(4.20,2.14)(5.20,2.83)(6.20,3.52)(7.20,4.53)(8.20,5.74)
(9.20,7.32)(10.20,9.52)}node[right=4pt]{{\normalsize\color{red}$d=10$}};
\draw[line width=1pt,color=blue,mark=o,mark size=3pt,draw=blue] plot
coordinates{(0.20,0.66)(1.20,0.73)(2.20,0.80)(3.20,0.87)
(4.20,0.93)(5.20,1.00)(6.20,1.06)(7.20,1.12)(8.20,1.17)
(9.20,1.22)(10.20,1.27)}node[right=4pt]{{\normalsize\color{blue}$d=3$}};
\draw[line width=1pt,color=green,mark=o,mark size=3pt,draw=green]
plot coordinates{(0.20,0.76)(1.20,0.90)(2.20,1.00)(3.20,1.16)
(4.20,1.33)(5.20,1.49)(6.20,1.66)(7.20,1.83)(8.20,2.05)
(9.20,2.20)(10.20,2.34)}node[right=4pt]{{\normalsize\color{green}$d=4$}};
\draw[line width=1pt,color=black,mark=o,mark size=3pt,draw=black]
plot coordinates{(0.20,0.81)(1.20,0.98)(2.20,1.18)(3.20,1.40)
(4.20,1.64)(5.20,1.89)(6.20,2.25)(7.20,2.54)(8.20,2.98)
(9.20,3.34)(10.20,3.73)}node[right=4pt]{{\normalsize\color{black}$d=5$}};
\draw[line width=1pt,color=red,mark=o,mark size=3pt,draw=red] plot
coordinates{(0.20,0.84)(1.20,1.06)(2.20,1.26)(3.20,1.54)
(4.20,1.86)(5.20,2.23)(6.20,2.69)(7.20,3.20)(8.20,3.77)
(9.20,4.43)(10.20,5.33)}node[right=4pt]{{\normalsize\color{red}$d=6$}};
\draw[line width=1pt,color=blue,mark=o,mark size=3pt,draw=blue] plot
coordinates{(0.20,0.85)(1.20,1.07)(2.20,1.32)(3.20,1.64)
(4.20,2.03)(5.20,2.48)(6.20,3.01)(7.20,3.70)(8.20,4.51)
(9.20,5.46)(10.20,6.53)}node[right=4pt]{{\normalsize\color{blue}$d=7$}};
\draw[line width=1pt,color=green,mark=o,mark size=3pt,draw=green]
plot coordinates{(0.20,0.81)(1.20,1.08)(2.20,1.37)(3.20,1.70)
(4.20,2.12)(5.20,2.67)(6.20,3.28)(7.20,4.13)(8.20,5.04)
(9.20,6.21)(10.20,7.71)}node[right=4pt]{{\normalsize\color{green}$d=8$}};
\draw[line width=1pt,color=black,mark=o,mark size=3pt,draw=black]
plot coordinates{(0.20,0.82)(1.20,1.07)(2.20,1.35)(3.20,1.74)
(4.20,2.24)(5.20,2.80)(6.20,3.52)(7.20,4.36)(8.20,5.56)
(9.20,6.87)(10.20,8.52)}node[right=4pt]{{\normalsize\color{black}$d=9$}};
\end{tikzpicture}
\caption{Simulation results of the total number of times the
procedure \textsf{Dominated} is called for in the first phase for
$d=3,4,5,\dots,10$ and $n=2^k$ for $k$ from $10 $ to $20$. Here we
plot $\frac{\sum_{i=1}^nT_i}{n\log_2n}$ against $k=\log _2n$.}
\label{effect}
\end{figure}
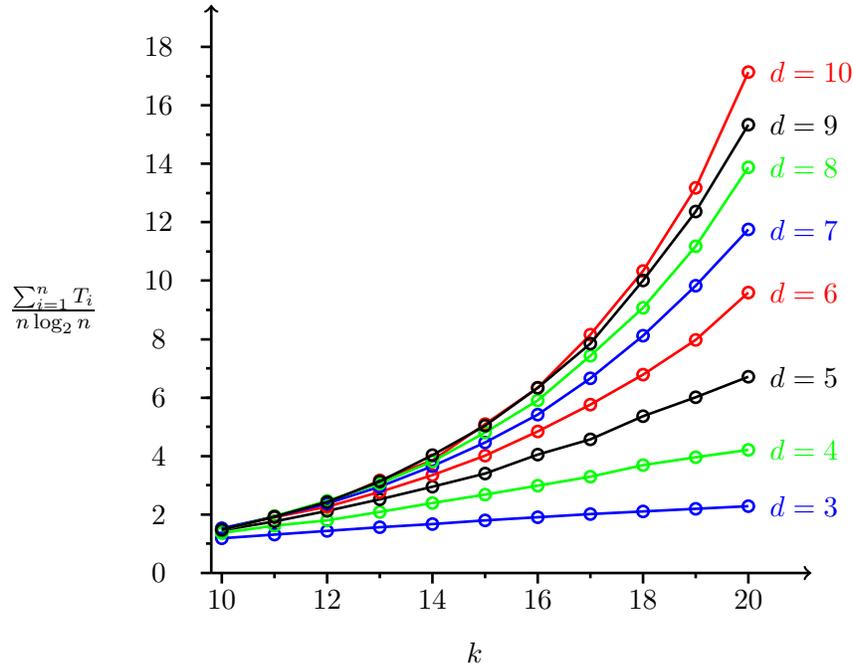

One way of seeing why our algorithm suffers less from the so-called
``curse of dimensionality" than other algorithms in such extreme
cases is as follows. As is obvious from the proof of
Theorem~\ref{thm-ec}, the time complexity is proportional to the
order of $|D_{\v{r}}|/ |A_{\v{r}}|$. The more slender $A_{\v{r}}$
is, the larger $|D_{\v{r}}|/ |A_{\v{r}}|$ becomes. All four possible
patterns of $A_{\v{r}}$ for $d=3$ are shown in Figure~\ref{d3}. The
slenderness does not seem to worsen rapidly as there is some sort of
counter-balancing process at play; see Figure~\ref{d3}.

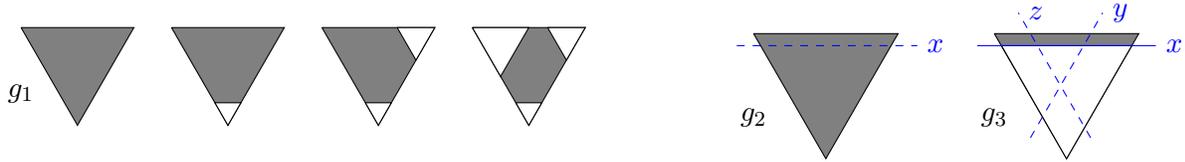
\begin{figure}[!ht]
\begin{minipage}{0.5\textwidth}
\centering
\begin{tikzpicture} [scale=0.25]
\filldraw[fill=gray] (-8,0) -- (-2,0) -- (-5,-5.2) -- (-8,0);
\draw (-8,-3.5) node {$g_1$};%
\filldraw[fill=gray] (0,0) -- (6,0) -- (3,-5.2) -- (0,0);
\filldraw[fill=white] (2.31,-4) -- (3.69,-4) -- (3,-5.2) --
(2.31,-4);%
\filldraw[fill=gray] (8,0) -- (14,0) -- (11,-5.2) -- (8,0);
\filldraw[fill=white] (10.31,-4) -- (11.69,-4) -- (11,-5.2)
--(10.31,-4);%
\filldraw[fill=white] (12,0) -- (14,0) -- (13,-1.732) -- (12,0);
\filldraw[fill=gray] (16,0) -- (22,0) -- (19,-5.2) -- (16,0);
\filldraw[fill=white] (18.31,-4) -- (19.69,-4) -- (19,-5.2)
--(18.31,-4);%
\filldraw[fill=white] (20,0) -- (22,0) -- (21,-1.732)
--(20,0); %
\filldraw[fill=white] (16,0) -- (19,0) -- (17.5,-2.598) --(16,0);
\end{tikzpicture}
\end{minipage}
\begin{minipage}{0.5\textwidth}
\centering
 \begin{tikzpicture} [scale=0.32]
\filldraw[fill=gray] (-10,0) -- (-4,0) -- (-7,-5.2) -- (-10,0);
\filldraw[fill=gray] (0,0) -- (6,0) -- (3,-5.2) -- (0,0);
\filldraw[fill=white](0.3,-0.5) -- (5.71,-0.5) -- (3,-5.2) --
(0.3,-0.5);%
\draw (-10,-3.5) node {$g_2$};%
\draw (0,-3.5) node {$g_3$}; \color{blue}%
\draw[dashed] (-10.7,-0.5) -- (-3.2,-0.5) node[right] {$x$};
\draw(-0.711,-0.5) -- (6.71,-0.5) node[right] {$x$};%
\draw[dashed] (1.5,-4.33) -- (4.5,0.866) node[right] {$y$};
\draw[dashed] (4,-4.33) -- (1,0.866) node[right] {$z$};
\end{tikzpicture}
\end{minipage}
\caption{Here $d=3$. All four possible configurations of $A_{\v{r}}$
are shown on the left (the four smaller, We can see how $A_{\v{r}}$
tends to keep from getting too slender by the interaction of
$x$-axis, $y$-axis and $z$-axis. Take the leftmost region (graph
$g_1$) for instance. Whenever $A_{\v{r}}$ is split less evenly by
$x$-axis (graph $g_2$), later splittings along $y$-axis or along
$z$-axis tend to counterbalance the effect caused by $x$-axis (graph
$g_3$).} \label{d3}
\end{figure}

\section{Applications}
\label{sec-app}

In this section, we apply algorithm \textsf{Maxima} to find
successively the maximal layers and to search for the longest common
subsequence of multiple sequences, respectively. In both cases, our
algorithms generally achieve better performance.

\subsection{Maximal layers}

The problem is to split the input set of points $\v{p}$ into layers
according to maxima. Let $\v{L}_k$ denote the $k$-th maximal layer
of $\v{p}$. Then $\v{L}_1=\v{Max}(\v{p})$ and
\[
    \v{L}_k:=\v{Max}
    \left(\v{p}\setminus\bigcup_{1\le i<k}
    \v{L} _i\right) ,\quad \text{for }k\ge 2.
\]

Maximal layers have been widely applied in multi-objective
optimization problems, and algorithms with $O(n \log n)$-time
complexity were known for finding the two- and three-dimensional
maximal layers; see \cite{BV08,BG04}.

By identifying the first few layers of maxima to preserve the
so-called elitism, Srinivas and Deb \cite{SD95} proposed a
multi-objective evolutionary algorithm, called non-dominated sorting
genetic algorithm (NSGA). This algorithm was later improved and
called NSGA-II \cite{Deb02}, which reduces the worst-case time
complexity from $O(dn^3) $ to $O(dn^2)$ and soon became extremely
popular. Omitting the details of the corresponding genetic
algorithms, the NSGA-II algorithm \cite{Deb02} for finding the
maximal layers can be extracted and summarized in the following two
steps.

\begin{description}

\item[Step 1:]  For each point $\v{p}_i$, compute the
number of points that dominate it $n_i :=\left| \{
\v{p}_j:\v{p}_i\prec \v{p}_j\}\right|$ ($n_i$ will be referred to as
the rank of the point $\v{p}_i$) and the set of points dominated by
it $\v{S}_i :=\{\v{p}_j:\v{p}_j\prec \v{p}_i\}$.

\item[Step 2:]  Then the maximal layers can be determined by $n_i$
and $\v{S}_i$ as follows. The first layer $\v{L}_1$ contains the
points with zero rank. For $k\ge 2$, remove $\v{L}_{k-1}$ and update
the rank $n_i$ by using $\v{S}_i$. Then, $\v{L}_k$ is the set of the
points with zero rank among all points that remain.
\end{description}

The running time is obviously $O(dn^2)$ since all pairs of points
are compared.

A straightforward way to compute the maximal layers is to find
successively the maxima after the removal of each layer.

\quad

\textbf{Algorithm} \textsf{Peeling}\\
\hspace*{1cm}//\textbf{Input}: A sequence of points
$\{\v{p}_1,\dots, \v{p}_{n}\}$\\
\hspace*{1cm}//\textbf{Output}: Maximal layers
$\v{L}_1,\v{L}_2,\dots$ \\%
\hspace*{1cm}\textbf{begin}\\
\hspace*{2cm} $k:=0;\v{q}:=\{\v{p}_1,\dots
,\v{p}_{n}\}$ \\
\hspace*{2cm} \textbf{while} ($\left| \v{q}\right| $ $>0$)\\
\hspace*{3cm}$k:=k+1$ \\
\hspace*{3cm}$\v{L}_k:=$\textsf{\ Find-Maxima} $(\v{q})$\\
\hspace*{3cm}$\v{q}:=\v{q}-\v{L}_k$ \\
\hspace*{1cm}\textbf{end}

\quad

Algorithm \textsf{Peeling} is simple and efficient in average
situations, even though the worst-case complexity is $O(n^3)$. Any
maxima-finding algorithm can be used for the procedure
$\textsf{Find-Maxima} (\v{q})$. To study the average behavior of
algorithm \textsf{Peeling}, we compare two procedures for
$\textsf{Find-Maxima}$: algorithm \textsf{Maxima} and algorithm
\textsf{Naive}. Algorithm \textsf{Naive} finds maxima using pairwise
comparisons.

\quad

\textbf{Algorithm} \textsf{Naive}\\
\hspace*{1cm}//\textbf{Input}: A set of points
$\v{q}=\{\v{q} _1,\dots ,\v{q}_{n}\}$\\
\hspace*{1cm}//\textbf{Output}: $\v{M} =
\v{Max}(\v{q})$. \\%
\hspace*{1cm} \textbf{begin} \\
\hspace*{2cm} $\v{M} := \{\}$ \\
\hspace*{2cm} \textbf{for} $i:=1$ \textbf{to} $n$ \textbf{do}\\
\hspace*{3cm} \textbf{for} $j:=1$ \textbf{to} $n$ \textbf{do}\\
\hspace*{4cm} \textbf{if} ($i\neq j$ and
$\v{q}_i \v{\prec q}_j$ ) \textbf{then} break\\
\hspace*{4cm} \textbf{if} ($j=n$) \textbf{then}
insert $\v{q}_i$ into $\v{M}$\\ %
\hspace*{1cm} \textbf{end}

\begin{thm} If $\v{p}_1,\dots ,\v{p}_n$ are independently
and uniformly sampled from any given region in $\mb{R}^d$, then the
expected running time of algorithm \textsf{Peeling} using algorithm
\textsf{Naive} is $O\left( n^2\log (K+1)\right)$, conditioned on the
number of maximal layers $K$.
\end{thm}
\begin{proof}
Consider the event that the total number of layers is $K$ and the
number of points in the $i$-th layer $\v{L}_i$ is $\ell_i$ for $1\le
i\le K$.

We now fix $k$. At the moment of computing $\v{L}_k$, the total
number of remaining points is equal to $N_k :=
\sum_{i=k}^{K}\ell_i$. If a point $\v{p}$ is in the $i$-th layer for
$i\ge k$, then the number of points that dominate $\v{p}$ is at
least $i-k$. Thus, the expected number of comparisons that $\v{p}$
involves in the loop for computing the $k$-th layer maxima is upper
bounded by
\[
    \le \left\{\begin{array}{ll}
        N_k, & \text{if }i=k, \\
        N_k/(i-k), & \text{if }i>k,
    \end{array}\right.
\]
since the remaining points preserve the randomness. Summing over all
$\v{p}$ and $k$, we obtain the upper bound for the expected number
of comparisons used
\begin{align*}
    \sum_{k=1}^K \ell_kN_k+\sum_{k=1}^K
    \sum_{i=k+1}^K \frac{\ell_i N_k}{i-k}
    &\le n^2 + n\sum_{i=2}^K\sum_{k=1}^{i-1}
    \frac{\ell_i}{i-k} \\
    &\le n^2 + n^2\left(1+\log K\right).
\end{align*}
This completes the proof.
\end{proof}
Note that the proof also extends to more general non-uniform
distributions.

We compare the numbers of scalar comparisons used by the following
three algorithms for finding the maximal layers: Deb et al.'s
algorithm \cite{Deb02}, algorithm \textsf{Peeling} using
\textsf{Maxima}, and algorithm \textsf{Peeling} using
\textsf{Naive}. The simulation results are shown in
Figure~\ref{layers}. Note that we reverse the order of the remainder
after a layer is found to make the algorithm more efficient. It is
clear that algorithm \textsf{Peeling} using \textsf{Maxima}
outperforms generally the other two, especially for higher
dimensional samples in large data sets.

\begin{figure}[tbp]
\begin{minipage}{0.3\textwidth}
\begin{tikzpicture}[scale=0.35]
\draw(-1.00,11.00)node{\tiny \#(comparisons)};
\draw(5.80,-1.50)node{\tiny$d$};
\draw(1.67,-1.50)node{\tiny$(n=100)$};%
\draw[line width=1pt](0,0)--(10.20,0);%
\foreach \x/\k in
{0.20/2,1.45/3,2.70/4,3.95/5,5.20/6,6.45/7,7.70/8,8.95/9,10.20/10}{
    \draw[line width=1pt](\x,-0.2)--(\x,0);
    \draw(\x,-0.5)circle(0pt)node[black]{\tiny$\k$};
}%
\foreach \x in
{0.45,0.70,0.95,1.20,1.70,1.95,2.20,2.45,2.95,3.20,3.45,3.70,4.20,
4.45,4.70,4.95,5.45,5.70,5.95,6.20,6.70,6.95,7.20,7.45,7.95,8.20,
8.45,8.70,9.20,9.45,9.70,9.95}
\draw[line width=1pt](\x,-0.1)--(\x,0);%
\draw[line width=1pt](0,0)--(0,10.00);%
\foreach \y/\k in
{0.00/0.0e+00,1.67/5.0e+03,3.33/1.0e+04,5.00/1.5e+04,
6.67/2.0e+04,8.33/2.5e+04,10.00/3.0e+04}{
    \draw[line width=1pt](-0.2,\y)--(0,\y);
    \draw(-2,\y)circle(0pt)node[black]{\tiny$\k$};
}%
\foreach \y in
{0.33,0.67,1.00,1.33,2.00,2.33,2.67,3.00,3.67,4.00,4.33,
4.67,5.00,5.33,5.67,6.00,6.33,6.67,7.00,7.33,7.67,8.00,
8.33,8.67,9.00,9.33,9.67}
\draw[line width=1pt](-0.1,\y)--(0,\y);%
\draw[line width=1pt,mark options={color=black},mark=o,mark
size=5pt,color=red] plot coordinates{(0.20,4.95)(1.45,5.78)
(2.70,6.24)(3.95,6.46)(5.20,6.62)(6.45,6.63)(7.70,6.64)
(8.95,6.69)(10.20,6.67)}node[above left]{{\tiny\color{red}Deb}};
\draw[line width=1pt,mark options={color=black},mark=pentagon,mark
size=5pt,color=blue] plot coordinates{(0.20,4.98)(1.45,4.73)
(2.70,4.92)(3.95,4.91)(5.20,5.13)(6.45,5.50)(7.70,5.78)
(8.95,6.14)(10.20,6.34)}node[below
left=6pt]{{\tiny\color{blue}\textsf{Naive}}};%
\draw[line width=1pt,mark options={color=black},mark=triangle,mark
size=5pt,color=green] plot coordinates{(0.20,2.66)(1.45,3.03)
(2.70,3.66)(3.95,4.46)(5.20,5.61)(6.45,6.30)(7.70,7.24)
(8.95,8.27)(10.20,8.95)}node[above
left]{{\tiny\color{green}\textsf{Maxima}}};
\end{tikzpicture}\end{minipage}
\hspace{0.3cm}
\begin{minipage}{0.3\textwidth}
\begin{tikzpicture}[scale=0.35]
\draw(-1.00,11.00)node{\tiny \#(comparisons)};
\draw(5.80,-1.50)node{\tiny$d$};
\draw(1.67,-1.50)node{\tiny$(n=1000)$};%
\draw[line width=1pt](0,0)--(10.20,0);%
\foreach \x/\k in
{0.20/2,1.45/3,2.70/4,3.95/5,5.20/6,6.45/7,7.70/8,8.95/9,10.20/10}{
    \draw[line width=1pt](\x,-0.2)--(\x,0);
    \draw(\x,-0.5)circle(0pt)node[black]{\tiny$\k$};
}%
\foreach \x in
{0.45,0.70,0.95,1.20,1.70,1.95,2.20,2.45,2.95,3.20,3.45,
3.70,4.20,4.45,4.70,4.95,5.45,5.70,5.95,6.20,6.70,6.95,
7.20,7.45,7.95,8.20,8.45,8.70,9.20,9.45,9.70,9.95}
\draw[line width=1pt](\x,-0.1)--(\x,0);%
\draw[line width=1pt](0,0)--(0,10.00);%
\foreach \y/\k in
{0.00/0.0e+00,2.50/5.0e+05,5.00/1.0e+06,7.50/1.5e+06,10.00/2.0e+06}{
    \draw[line width=1pt](-0.2,\y)--(0,\y);
    \draw(-2,\y)circle(0pt)node[black]{\tiny$\k$};
}%
\foreach \y in
{0.50,1.00,1.50,2.00,3.00,3.50,4.00,4.50,5.50,6.00,6.50,
7.00,8.00,8.50,9.00,9.50}
\draw[line width=1pt](-0.1,\y)--(0,\y);%
\draw[line width=1pt,mark options={color=black},mark=o,mark
size=5pt,color=red] plot coordinates{(0.20,7.49)(1.45,8.74)
(2.70,9.36)(3.95,9.67)(5.20,9.83)(6.45,9.91)(7.70,9.96)
(8.95,9.96)(10.20,9.98)}node[above
left]{{\tiny\color{red}Deb}};%
\draw[line width=1pt,mark options={color=black},mark=pentagon,mark
size=5pt,color=blue] plot coordinates{(0.20,7.34)(1.45,6.66)
(2.70,6.76)(3.95,6.92)(5.20,7.31)(6.45,7.58)(7.70,7.87)
(8.95,8.19)(10.20,8.55)}node[below
left=6pt]{{\tiny\color{blue}\textsf{Naive}}};%
\draw[line width=1pt,mark options={color=black},mark=triangle,mark
size=5pt,color=green] plot coordinates{(0.20,1.58)(1.45,1.54)
(2.70,1.76)(3.95,2.14)(5.20,2.56)(6.45,2.92)(7.70,3.43)
(8.95,3.85)(10.20,4.33)}node[above
left]{{\tiny\color{green}\textsf{Maxima}}};
\end{tikzpicture}\end{minipage}
\hspace{0.3cm}
\begin{minipage}{0.3\textwidth}
\begin{tikzpicture}[scale=0.35]
\draw(-1.00,11.00)node{\tiny \#(comparisons)};
\draw(5.80,-1.50)node{\tiny$d$};
\draw(1.67,-1.50)node{\tiny$(n=10000)$};%
\draw[line width=1pt](0,0)--(10.20,0);%
\foreach \x/\k in
{0.20/2,1.45/3,2.70/4,3.95/5,5.20/6,6.45/7,7.70/8,8.95/9,10.20/10}{
    \draw[line width=1pt](\x,-0.2)--(\x,0);
    \draw(\x,-0.5)circle(0pt)node[black]{\tiny$\k$};
}%
\foreach \x in {0.45,0.70,0.95,1.20,1.70,1.95,2.20,2.45,2.95,3.20,
3.45,3.70,4.20,4.45,4.70,4.95,5.45,5.70,5.95,6.20,
6.70,6.95,7.20,7.45,7.95,8.20,8.45,8.70,9.20,9.45,9.70,9.95}
\draw[line width=1pt](\x,-0.1)--(\x,0);%
\draw[line width=1pt](0,0)--(0,10.00);%
\foreach \y/\k in
{0.00/0.0e+00,2.50/5.0e+07,5.00/1.0e+08,7.50/1.5e+08,10.00/2.0e+08}{
    \draw[line width=1pt](-0.2,\y)--(0,\y);
    \draw(-2,\y)circle(0pt)node[black]{\tiny$\k$};
}%
\foreach \y in {0.50,1.00,1.50,2.00,3.00,3.50,4.00,4.50,5.50,6.00,
6.50,7.00,8.00,8.50,9.00,9.50}
\draw[line width=1pt](-0.1,\y)--(0,\y);%
\draw[line width=1pt,mark options={color=black},mark=o,mark
size=5pt,color=red] plot coordinates{(0.20,7.50)(1.45,8.75)
(2.70,9.37)(3.95,9.69)(5.20,9.84)(6.45,9.92)(7.70,9.96)
(8.95,9.98)(10.20,9.99)}node[above left]{{\tiny\color{red}Deb}};
\draw[line width=1pt,mark options={color=black},mark=pentagon,mark
size=5pt,color=blue] plot coordinates{(0.20,7.25)(1.45,6.23)
(2.70,6.17)(3.95,6.36)(5.20,6.65)(6.45,6.92)(7.70,7.20)
(8.95,7.48)(10.20,7.73)}node[below
left=6pt]{{\tiny\color{blue}\textsf{Naive}}};%
\draw[line width=1pt,mark options={color=black},mark=triangle,mark
size=5pt,color=green] plot coordinates{(0.20,0.48)(1.45,0.38)
(2.70,0.42)(3.95,0.50)(5.20,0.61)(6.45,0.75)(7.70,0.89)
(8.95,1.03)(10.20,1.20)}node[above
left]{{\tiny\color{green}\textsf{Maxima}}};
\end{tikzpicture}\end{minipage}
\caption{Simulation of Deb's algorithm, and the peeling method with
algorithm \textsf{Naive} and algorithm \textsf{Maxima},
respectively. We compare the number of scalar comparisons used in
the algorithms. Here the sample size $n=10^2,10^3,10^4$ and the
points are generated uniformly from $[0,1]^d$ for $d=2,3,\dots
,10$.} \label{layers}
\end{figure}
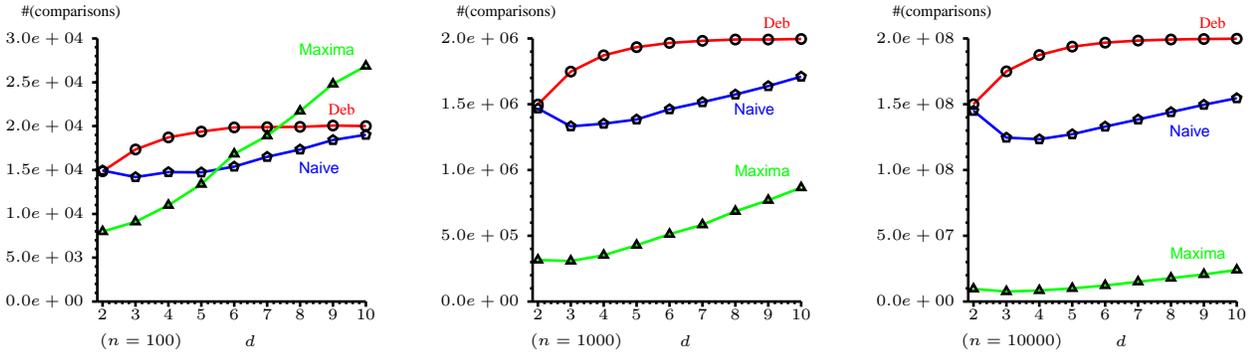

\subsection{The multiple longest common subsequence problem}

Given two or more strings (or sequences), the longest common
subsequence (LCS for short) problem is to determine the longest
common subsequence obtained by removing zero or more symbols from
each string. For example, if $\v{s}_1=aabbc$ and $\v{s}_2=abac$ then
$LCS(\v{s}_1,\v{s}_2)$, the LCS of $\v{s}_1$ and $\v{s}_2$, is
$abc$. The LCS of sequences is widely used in computational biology,
notably in DNA and protein sequence analysis.

Various algorithms for computing an LCS between two strings were
derived in the literature, but much fewer algorithms are devoted to
the LCS of more than two strings. Hakata and Imai \cite{HI98}
proposed a method for solving efficiently the multiple LCS problem.
The method is essentially based on minima-finding.

Let $\v{s}_1=a_1a_2\cdots a_{n}$ and $\v{s}_2=b_1b_2\cdots b_{m}$ be
two strings. We say that $(i,j)$ is a \emph{match} if $a_i=b_j$.
Given two matches $(i_1,j_1)$ and $(i_2,j_2)$. If $i_1<i_2$ and
$j_1<j_2$ then
\[
    \text{LCS}(a_1\cdots a_{i_1},b_1\cdots b_{j_1})
    <\text{LCS}(a_1\cdots a_{i_2},b_1\cdots b_{j_2}).
\]
Thus, finding the LCS can be roughly regarded as finding the maximal
layers of all possible matches. However, the number of matches is
usually too large. The approach proposed in \cite{HI98} is to find
the layers one after another as follows. Assume we have found the
$k$-th layer, $C_k$, then the $(k+1)$-st layer is the minima of all
successors of $C_k$, where a match $(i_2,j_2)$ is called a
\emph{successor} of another match $(i_1,j_1) $ if $i_1<i_2$ and
$j_1<j_2$ and there is no match between them. The minima-finding
algorithm proposed in \cite{HI98} is an improvement over algorithm
\textsf{Naive}. The algorithm runs as follows.

\quad

\textbf{Algorithm} \textsf{Hakata-Imai} \\
\hspace*{1cm}//\textbf{Input}: A set of points
$\v{q}=\{\v{q}_1,\dots ,\v{q}_{n}\}$\\
\hspace*{1cm}//\textbf{Output}: $\v{M}$ contains minima of
$\v{q}$ \\%
\hspace*{1cm}\textbf{begin} \\
\hspace*{2cm} $\v{M} := \{\}$ \\
\hspace*{2cm} \textbf{for} $i:=1$ \textbf{to} $n$ \textbf{do} \\
\hspace*{3cm} \textbf{if} $\v{q}_i$ is unmarked \textbf{then} \\
\hspace*{4cm} \textbf{for} $j:=1$ \textbf{to} $n$ \textbf{do} \\
\hspace*{5cm} \textbf{if} $\v{q}_j$ is unmarked \textbf{then} \\
\hspace*{6cm} \textbf{if} ($\v{q}_i\v{\prec q}_j$)
\textbf{then } mark $\v{q}_j$ \\%
\hspace*{6cm} \textbf{if} ($\v{q}_j\v{\prec q}_i$)
\textbf{then } mark $\v{q}_i$ \\%
\hspace*{4cm} \textbf{if} $\v{q}_i$ is unmarked \textbf{then}
insert $\v{q}_i$ into $\v{M}$\\
\hspace*{1cm}
\textbf{end}

\quad

This algorithm is similar to the list algorithm if we consider
node-marking as a substitute of node-deletion.

We compare the performance of \textsf{Hakata-Imai} and
\textsf{Maxima} for the number of strings $3,5,7$ and alphabet sizes
$4,20$. See the experimental results in Figure \ref{lcs} where the
improvement achieved by our algorithm is visible.

\begin{figure}[tbp]
\begin{minipage}{0.3\textwidth}
\begin{tikzpicture}[scale=0.35]
\draw(-2.00,5.00)node{\tiny ratio};
\draw(6.67,-1.50)node{\tiny$n$};%
\draw(1.67,-1.50)node{\tiny($d=3$, $s=4$)};%
\draw[line width=1pt](0,0)--(10.20,0);%
\foreach \x/\k in {0.20/200,2.70/275,5.20/350,7.70/425,10.20/500} {
    \draw[line width=1pt](\x,-0.2)--(\x,0);
    \draw(\x,-0.5)circle(0pt)node[black]{\tiny$\k$};
}%
\draw[line width=1pt](0,0)--(0,10.00);%
\foreach \y/\k in {0.00/0,3.33/5,6.67/10,10.00/15} {%
    \draw[line width=1pt](-0.2,\y)--(0,\y);%
    \draw(-1,\y)circle(0pt)node[black]{\tiny$\k$};
}%
\foreach \y in
{0.67,1.33,2.00,2.67,4.00,4.67,5.33,6.00,7.33,8.00,8.67,9.33}
\draw[line width=1pt](-0.1,\y)--(0,\y);%
\draw[line width=1pt,color=red,mark=o,mark size=3pt,mark
options={color=black}] plot
coordinates{(0.20,3.22)(0.87,3.51)(1.53,3.87)(2.20,4.42)(2.87,4.66)
(3.53,5.10)(4.20,5.53)(4.87,6.07)(5.53,6.58)(6.20,7.22)(6.87,7.68)
(7.53,8.24)(8.20,8.62)(8.87,9.05)(9.53,9.66)(10.20,10.40)};
\end{tikzpicture}\end{minipage}
\hspace{0.3cm}
\begin{minipage}{0.3\textwidth}
\begin{tikzpicture}[scale=0.35]
\draw(-2.00,5.00)node{\tiny ratio};
\draw(6.67,-1.50)node{\tiny$n$};%
\draw(1.67,-1.50)node{\tiny($d=5$, $s=4$)};%
\draw[line width=1pt](0,0)--(10.20,0);%
\foreach \x/\k in {0.20/60,2.70/75,5.20/90,7.70/105,10.20/120} {
    \draw[line width=1pt](\x,-0.2)--(\x,0);%
    \draw(\x,-0.5)circle(0pt)node[black]{\tiny$\k$};%
}%
\draw[line width=1pt](0,0)--(0,10.00);%
\foreach \y/\k in {0.00/0,3.33/5,6.67/10,10.00/15} {
    \draw[line width=1pt](-0.2,\y)--(0,\y);
    \draw(-1,\y)circle(0pt)node[black]{\tiny$\k$};
}%
\foreach \y in
{0.67,1.33,2.00,2.67,4.00,4.67,5.33,6.00,7.33,8.00,8.67,9.33}
\draw[line width=1pt](-0.1,\y)--(0,\y);%
\draw[line width=1pt,color=red,mark=o,mark size=3pt,mark
options={color=black}] plot
coordinates{(0.20,2.09)(1.87,2.62)(3.53,3.20)(5.20,4.50)
(6.87,5.81)(8.53,6.82)(10.20,8.97)};
\end{tikzpicture}\end{minipage}
\hspace{0.3cm}
\begin{minipage}{0.3\textwidth}
\begin{tikzpicture}[scale=0.35]
\draw(-2.00,5.00)node{\tiny ratio};%
\draw(6.67,-1.50)node{\tiny$n$};%
\draw(1.67,-1.50)node{\tiny($d=7$, $s=4$)};%
\draw[line width=1pt](0,0)--(10.20,0);%
\foreach \x/\k in {0.20/50,2.70/58,5.20/65,7.70/72,10.20/80} {
    \draw[line width=1pt](\x,-0.2)--(\x,0);
    \draw(\x,-0.5)circle(0pt)node[black]{\tiny$\k$};%
}%
\draw[line width=1pt](0,0)--(0,10.00);%
\foreach \y/\k in {0.00/0,3.33/5,6.67/10,10.00/15} {
    \draw[line width=1pt](-0.2,\y)--(0,\y);
    \draw(-1,\y)circle(0pt)node[black]{\tiny$\k$};
}%
\foreach \y in
{0.67,1.33,2.00,2.67,4.00,4.67,5.33,6.00,7.33,8.00,8.67,9.33}
\draw[line width=1pt](-0.1,\y)--(0,\y);%
\draw[line width=1pt,color=red,mark=o,mark size=3pt,mark
options={color=black}] plot coordinates{(0.20,2.49)(1.87,2.94)
(3.53,3.54)(5.20,4.61)(6.87,5.86)(8.53,7.30)(10.20,9.12)};
\end{tikzpicture}\end{minipage}
\begin{minipage}{0.3\textwidth}
\begin{tikzpicture}[scale=0.35]
\draw(-2.00,5.00)node{\tiny ratio};
\draw(6.67,-1.50)node{\tiny$n$};%
\draw(1.67,-1.50)node{\tiny($d=3$, $s=20$)};%
\draw[line width=1pt](0,0)--(10.20,0);%
\foreach \x/\k in {0.20/300,2.70/400,5.20/500,7.70/600,10.20/700} {
    \draw[line width=1pt](\x,-0.2)--(\x,0);
    \draw(\x,-0.5)circle(0pt)node[black]{\tiny$\k$};
}%
\draw[line width=1pt](0,0)--(0,10.00);%
\foreach \y/\k in {0.00/0,3.33/5,6.67/10,10.00/15}{
    \draw[line width=1pt](-0.2,\y)--(0,\y);
    \draw(-1,\y)circle(0pt)node[black]{\tiny$\k$};
}%
\foreach \y in
{0.67,1.33,2.00,2.67,4.00,4.67,5.33,6.00,7.33,8.00,8.67,9.33}
\draw[line width=1pt](-0.1,\y)--(0,\y);%
\draw[line width=1pt,color=red,mark=o,mark size=3pt,mark
options={color=black}] plot coordinates{(0.20,2.17)(0.70,2.30)
(1.20,2.45)(1.70,2.64)(2.20,2.81)(2.70,2.98)(3.20,3.12)(3.70,3.32)
(4.20,3.54)(4.70,3.62)(5.20,3.84)(5.70,4.09)(6.20,4.23)(6.70,4.45)
(7.20,4.66)(7.70,4.93)(8.20,5.09)(8.70,5.27)(9.20,5.52)(9.70,5.74)
(10.20,5.99)};
\end{tikzpicture}\end{minipage}
\hspace{0.3cm}
\begin{minipage}{0.3\textwidth}
\begin{tikzpicture}[scale=0.35]
\draw(-2.00,5.00)node{\tiny ratio}; \draw(6.67,-1.50)node{\tiny$n$};
\draw(1.67,-1.50)node{\tiny($d=5$, $s=20$)};%
\draw[line width=1pt](0,0)--(10.20,0);%
\foreach \x/\k in{0.20/120,2.70/158,5.20/195,7.70/232,10.20/270}{
    \draw[line width=1pt](\x,-0.2)--(\x,0);
    \draw(\x,-0.5)circle(0pt)node[black]{\tiny$\k$};
}%
\draw[line width=1pt](0,0)--(0,10.00);%
\foreach \y/\k in {0.00/0,3.33/5,6.67/10,10.00/15} {
    \draw[line width=1pt](-0.2,\y)--(0,\y);
    \draw(-1,\y)circle(0pt)node[black]{\tiny$\k$};
}%
\foreach \y in {0.67,1.33,2.00,2.67,4.00,4.67,5.33,
6.00,7.33,8.00,8.67,9.33}%
\draw[line width=1pt](-0.1,\y)--(0,\y);%
\draw[line width=1pt,color=red,mark=o,mark size=3pt, mark
options={color=black}] plot coordinates{(0.20,1.46)
(0.87,1.71)(1.53,2.01)(2.20,2.22)(2.87,2.56)(3.53,2.85)
(4.20,3.24)(4.87,3.68)(5.53,4.18)(6.20,4.87)(6.87,5.19)
(7.53,6.06)(8.20,6.67)(8.87,7.28)(9.53,8.04)(10.20,9.03)};
\end{tikzpicture}\end{minipage}
\hspace{0.3cm}
\begin{minipage}{0.3\textwidth}
\begin{tikzpicture}[scale=0.35]
\draw(-2.00,5.00)node{\tiny ratio};%
\draw(6.67,-1.50)node{\tiny$n$};%
\draw(1.67,-1.50)node{\tiny($d=7$, $s=20$)};%
\draw[line width=1pt](0,0)--(10.20,0); \foreach \x/\k in
{0.20/100,2.70/120,5.20/140,7.70/160,10.20/180}{
    \draw[line width=1pt](\x,-0.2)--(\x,0);
    \draw(\x,-0.5)circle(0pt)node[black]{\tiny$\k$};}%
\draw[line width=1pt](0,0)--(0,10.00); \foreach \y/\k in
{0.00/0,3.33/5,6.67/10,10.00/15} {
    \draw[line width=1pt](-0.2,\y)--(0,\y);
    \draw(-1,\y)circle(0pt)node[black]{\tiny$\k$};}
\foreach \y in {0.67,1.33,2.00,2.67,4.00,4.67,5.33,
6.00,7.33,8.00,8.67,9.33} %
\draw[line width=1pt](-0.1,\y)--(0,\y);%
\draw[line width=1pt,color=red,mark=o,mark size=3pt, mark
options={color=black}] plot coordinates{(0.20,1.02)
(1.45,1.37)(2.70,1.73)(3.95,2.62)(5.20,2.88)(6.45,3.92)
(7.70,4.94)(8.95,5.95)(10.20,7.10)};
\end{tikzpicture}\end{minipage}
\caption{A plot of the ratio between the running time of
\textsf{Hakata-Imai} \cite{HI98} and that of \textsf{Maxima} when
the numbers of strings $d=3,5,7$, the alphabet size $s=4, 20$, and
$n$ is the length of the strings. All strings are uniformly
generated at random.} \label{lcs}
\end{figure}
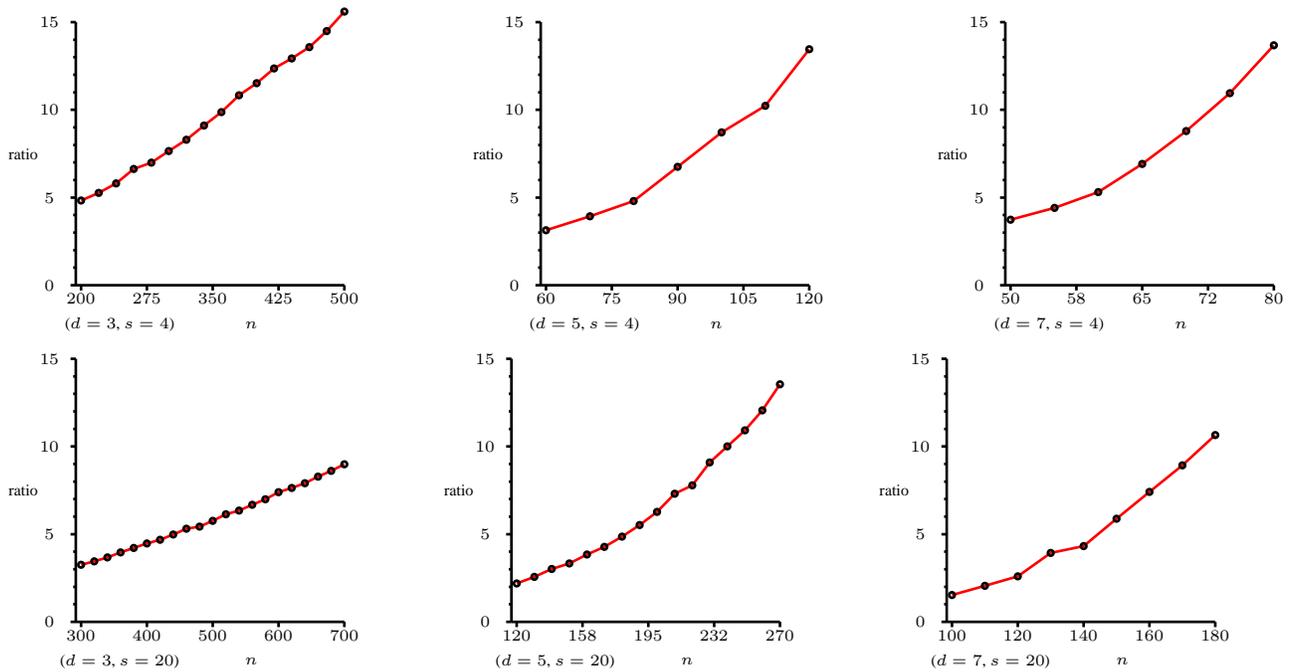


\begin{thebibliography}{99}

\bibitem{BDHT05} Z.-D. Bai, L. Devroye, H.-K. Hwang and T.-H. Tsai,
Maxima in hypercubes, \emph{Random Structures and Algorithms},
\textbf{27} (2005), 290--309.

\bibitem{BHLT01} Z.-D. Bai, H.-K. Hwang, W.-Q. Liang and T.-H. Tsai,
Limit theorems for the number of maxima in random samples from
planar regions, \emph{Electronic Journal of Probability}, \textbf{6}
(2001), paper no. 3, 41 pages.

\bibitem{BLP06} Z. D. Bai, S. Lee and M. D. Penrose, Rooted edges of
a minimal directed spanning tree on random points. \emph{Advances in
Applied Probability}, \textbf{38} (2006), 1--30.

\bibitem{BDJ99} J. Baik, P. Deift and K. Johansson, On the
distribution of the length of the longest increasing subsequence of
random permutations, \emph{Journal of the American Mathematical
Society}, \textbf{12} (1999), 1119--1178.

\bibitem{BCP08} I. Bartolini, P. Ciaccia and M. Patella, Efficient
sort-based skyline evaluation, \emph{ACM Transactions on Database
Systems}, \textbf{33} (2008), Article 31, 49 pages.

\bibitem{Baryshnikov07} Yu.\ Baryshnikov, On expected number of
maximal points in polytopes, 2007 Conference on Analysis of
Algorithms, \emph{DMTCS Proc.} AH, 2007, 227--236.

\bibitem{Bentley75} J. L. Bentley, Multidimensional binary search
trees used for associative searching, \emph{Communications of the
ACM}, \textbf{18} (1975),  509--517.

\bibitem{Bentley80} J. L. Bentley, Multidimensional
divide-and-conquer, \emph{Communications of the ACM}, \textbf{23}
(1980),  214--229.

\bibitem{BCL93} J. L. Bentley, K. L. Clarkson and D. B. Levine, Fast
linear expected-time algorithms for computing maxima and convex
hulls, \emph{Algorithmica}, \textbf{9} (1993), 168--183.

\bibitem{BR04} A. G. Bhatt and R. Roy, On a random directed spanning
tree. \emph{Advances in Applied Probability}, \textbf{36} (2004),
19--42.

\bibitem{BD08} G. Biau and L. Devroye, On the layered nearest
neighbour estimate, the bagged nearest neighbour estimate and the
random forest method in regression and classification, preprint,
(2008).

\bibitem{BV08} H. Blunck and J. Vahrenhold, In-place algorithms for
computing (layers of) maxima, \emph{Algorithmica}, to appear.

\bibitem{BW88} B. Bollob\'as and P. Winkler, The longest chain among
random points in Euclidean space, \emph{Proceedings of the American
Mathematical Society}, \textbf{103} (1988), 347--353.

\bibitem{BKS01} S. B\"{o}rzs\"{o}nyi, D. Kossmann and K. Stocker. The
skyline operator, \emph{Proceedings 17th International Conference on
Data Engineering}, pp.\ 421--430, 2001.

\bibitem{BG04} A. L. Buchsbaum and M. T. Goodrich, Three-dimensional
layers of maxima, \emph{Algorithmica}, \textbf{39} (2004), 275--286.

\bibitem{CHT03} W.-M. Chen, H.-K. Hwang and T.-H. Tsai, Efficient
maxima-finding algorithms for random planar samples, \emph{Discrete
Mathematics and Theoretical Computer Science}, \textbf{6} (2003),
107--122.

\bibitem{CL07} W.-M. Chen and W.-T. Lee, An efficient evolutionary
algorithm for multiobjective optimization problems, in \emph{IEEE
Pacific Rim Conference on Communications, Computers and Signal
Processing}, 2007, pp.\ 30--33.

\bibitem{CC06} C. A. Coello Coello, Evolutionary multi-objective
optimalization: a historical view the field, \emph{IEEE
Computational Intelligence Magazine}, February 2006, pp.\ 28--36.

\bibitem{CVL07} C. A. Coello Coello, D. A. Van Veldhuizen and G. B.
Lamont, \emph{Evolutionary Algorithms for Solving Multi-objective
Problems}, 2nd Ed., Springer, New York, 2007.

\bibitem{Deb01} K. Deb, \emph{Multi-Objective Optimization using
Evolutionary Algorithms}, John Wiley \& Sons, 2001.

\bibitem{Deb02} K. Deb, A. Pratap, S. Agarwal and T. Meyarivan, A
fast and elitist multiobjective genetic algorithm: NSGA-II,
\emph{IEEE Transactions on Evolutionary Computation}, \textbf{6}
(2002), 182--197.

\bibitem{Devroye83} L. Devroye, Moment inequalities for random
variables in computational geometry, \emph{Computing}, \textbf{30}
(1983), 111--119.

\bibitem{Devroye93} L. Devroye, Records, the maximal layer, and
uniform distributions in monotone sets,  \emph{Computers and
Mathematics with Applications}, \textbf{25} (1993), 19--31.

\bibitem{Devroye94} L. Devroye, On random Cartesian trees,
\emph{Random Structures and Algorithms}, \textbf{5} (1994),
305--327.

\bibitem{Devroye99} L. Devroye, A note on the expected time for
finding maxima by list algorithms, \emph{Algorithmica}, \textbf{23}
(1999), 97--108.

\bibitem{Ehrgott00} M. Ehrgott, \emph{Multicriteria Optimization},
Berlin, Springer, 2000.

\bibitem{FES03} J. Fieldsend, R.M. Everson and S. Singh, Using
unconstrained elite archives for multi-objective optimisation,
\emph{IEEE Transactions on Evolutionary Computation}, \textbf{7}
(2003), 305--323.

\bibitem{FG93} P. Flajolet and M. Golin, Exact asymptotics of
divide-and-conquer recurrences, \emph{Lecture Notes in Computer
Science}, \textbf{700}, pp.\ 137--149, Springer, Berlin, 1993.

\bibitem{GBT84} H. N. Gabow, J. L. Bentley and R. E. Tarjan, Scaling
and related techniques for geometry problems, \emph{Proceedings of
the 16th Annual ACM Symposium on Theory of Computing},
pp.\ 135--143, 1984.

\bibitem{Gnedin07} A. V. Gnedin, The chain records, \emph{Electronic
Journal of Probability}, \textbf{12} (2007), 767--786 (electronic).

\bibitem{GSG07} P. Godfrey, R. Shipley and J. Gryz, Algorithms and
analysis for maximal vector computation, \emph{The VLDB Journal},
\textbf{16} (2007), 5--28.

\bibitem{G89} D. E. Goldberg, \emph{Genetic Algorithms in Search,
Optimization and Machine Learning}, Addison-Wesley Publishing
Company, Reading, Massachusetts, 1989.

\bibitem{Golin93} M. J. Golin, Maxima in convex regions, in
\emph{Proceedings of the Fourth Annual ACM-SIAM Symposium on
Discrete Algorithms} (Austin, TX, 1993),  352--360, ACM, New York,
1993.

\bibitem{Golin94} M. J. Golin, A provably fast linear-expected-time
maxima-finding algorithm, \emph{Algorithmica}, \textbf{11} (1994),
501--524.

\bibitem{Habenicht83} W. Habenicht, Quad trees, a datastructure for
discrete vector optimization problems, in \emph{Essays and Surveys
on Multiple Criteria Decision Making: Proceedings on the Fifth
International Conference on Multiple Criteria Decision Making,
1982}, pp.\ 136--145, Springer (1983).

\bibitem{HI98} K. Hakata and H. Imai, Algorithms for the longest
common subsequence problem for multiple strings based on geometric
maxima, \emph{Optimization Methods and Software}, \textbf{10}
(1998), 233--260.

\bibitem{Jensen03} M. Jensen, Reducing the run-time complexity of
multiobjective EAs: The NSGA-II and other algorithms, \emph{IEEE
Transactions on Evolutionary Computation}, \textbf{7} (2003),
503--515.

\bibitem{Kaldewaij87} A. Kaldewaij, Some algorithms based on the dual
of Dilworth's theorem, \emph{Science of Computer Programming},
\textbf{9} (1987),  85--89.

\bibitem{KS85} D. G. Kirkpatrick abd R. Seidel, Output-size sensitive
algorithms for finding maximal vectors, in \emph{Proceedings of the
first Annual Symposium on Computational Geometry}, 1985, 89--96.

\bibitem{KC00} J. D. Knowles and D.W. Corne, Approximating the
nondominated front using the Pareto archived evolution strategy,
\emph{Evolutionary Computation}, \textbf{8} (2000), 149--172.

\bibitem{Knuth97} D. E. Knuth, \emph{The Art of Computer Programming,
Volume 1: Fundamental Algorithms}, Third Edition, Addison-Wesley,
Reading, Massachusetts, 1997.

\bibitem{KRR02} D. Kossmann, F. Ramsak and S. Rost, Shooting stars in
the sky: An online algorithm for skyline queries, \emph{Proceedings
of the 28th International Conference on Very Large Data Bases}, pp.
275--286, 2002.

\bibitem{KLP75} H. T. Kung, F. Luccio and F. P. Preparata, On finding
the maxima of a set of vectors, \emph{Journal of the ACM},
\textbf{22} (1975), 469--476.

\bibitem{MTT02} S. Mostaghim, J. Teich and A. Tyagi, Comparison of
data structures for storing Pareto sets in MOEAs, \emph{Proceedings
World Congress on Computational Intelligence}, IEEE Press, pp.
843--849, 2002.

\bibitem{PTFS05} D. Papadias, Y. Tao, G. Fu and B. Seeger,
Progressive skyline computation in database systems, \emph{ACM
Transactions on Database Systems}, \textbf{30} (2005), 41--82.

\bibitem{PS85} F. P. Preparata and M. I. Shamos, \emph{Computational
Geometry. An Introduction.} Springer-Verlag, New York, 1985.

\bibitem{Schutze03} O. Sch\"{u}tze, A new data structure for the
nondominance problem in multiobjective optimization, in
\emph{Evolutionary Multicriterion Optimization}, Edited by C. M.
Fonseca, P. J. Fleming, E. Zitzler, K. Deb, and L. Thiele, Lecture
Notes in Computer Science, Vol.\ 2632, Springer, Berlin, Germany,
pp. 509--518, 2003.

\bibitem{SD95} N. Srinivas and K. Deb, Multiobjective function
optimization using nondominated sorting genetic algorithms,
\emph{Evolutionary Computation}, \textbf{2} (1995), 221--248.

\bibitem{SS96} M. Sun and R. E. Steuer, Quad-trees and linear lists
for identifying nondominated criterion vectors, \emph{INFORMS
Journal on Computing}, \textbf{8} (1996), 367--375.

\bibitem{TEO01} K. Tan, P. Eng and B. Ooi, Efficient progressive
skyline computation, \emph{Proceedings of the 27th International
Conference on Very Large Data Bases}, pp. 301--310, 2001.

\bibitem{ZDT00} E. Zitzler, K. Deb, and L. Thiele, Comparison of
multiobjective evolutionary algorithms: Empirical results, {\em
Evolutionary Computation}, \textbf{8} (2000), 173 -- 195.

\bibitem{ZT99} E. Zitzler and L. Thiele, Multiobjective evolutionary
algorithms: A comparative case study and the strength Pareto
approach, \emph{IEEE Transactions on Evolutionary Computation},
\textbf{3} (1999), 257-- 271.
\end{thebibliography}
\end{document}